\documentclass[11pt]{article}
\usepackage[margin=1in]{geometry}

\usepackage[utf8]{inputenc}
\usepackage[T1]{fontenc}

\usepackage{amsmath,amssymb,amsthm,amsfonts}
\usepackage{mathtools}
\usepackage{graphicx}
\usepackage{float}      
\usepackage{subcaption}
\usepackage{arydshln}
\usepackage{authblk}
\usepackage{multirow}
\usepackage{bm}
\usepackage{mathrsfs}   
\usepackage{setspace}
\usepackage{footnote}
\makesavenoteenv{table}

\usepackage[dvipsnames]{xcolor}
\usepackage{hyperref}
\hypersetup{
    colorlinks,
    linkcolor={blue!80!black},
    urlcolor={blue!80!black},
    citecolor=teal
}
\usepackage[capitalize]{cleveref}
\usepackage{url}
\usepackage[title]{appendix}

\usepackage{thmtools}
\usepackage{thm-restate}

\usepackage{quiver}
\usepackage{tikz}  
\usepackage{tikz-cd}  
\usetikzlibrary{positioning}
\usepackage{tikz-qtree}
\usepackage[all]{xy}
\usepackage{mathdots}       
\usepackage{caption}

\usepackage{qcircuit}
\newcommand{\ket}[1]{| #1 \rangle}
\newcommand{\bra}[1]{\langle #1|}

\newcommand{\braket}[2]{\langle #1 | #2 \rangle}

\DeclareMathOperator{\polylog}{polylog}
\DeclareMathOperator{\tr}{Tr}
\DeclareMathOperator{\rank}{rank}

\DeclareMathOperator{\supp}{supp}

\newcommand{\nc}{\newcommand}

\newcommand{\mb}{\mathbb}

\nc{\bco}[1]{\{0,1\}^{#1}}

\newcommand{\eps}{\varepsilon}

\allowdisplaybreaks

\newtheorem{thm}{Theorem}[section]
\newtheorem{prop}[thm]{Proposition}
\newtheorem{cor}[thm]{Corollary}
\newtheorem{lem}[thm]{Lemma}

\theoremstyle{definition}

\newtheorem{defn}[thm]{Definition}
\newtheorem{rmk}[thm]{Remark}

\numberwithin{equation}{section}

\newtheorem*{thm*}{Theorem}
\newtheorem*{lem*}{Lemma}
\newtheorem*{prop*}{Proposition}
\newtheorem*{defn*}{Definition}
\newtheorem*{prob*}{Problem}
\newtheorem*{ques*}{Question}

\usepackage{algorithm}
\usepackage{algorithmic}



\newcommand{\be}{\begin{equation}}
\newcommand{\ee}{\end{equation}}
\newcommand{\bes}{\begin{equation*}}
\newcommand{\ees}{\end{equation*}}
\newcommand{\bea}{\begin{eqnarray}}
\newcommand{\eea}{\end{eqnarray}}
\newcommand{\beas}{\begin{eqnarray*}}
\newcommand{\eeas}{\end{eqnarray*}}
\newcommand{\bal}{\begin{aligned}}
\newcommand{\eal}{\end{aligned}}

\usepackage{authblk}

\title{ Quantum space-depth tradeoffs for coherent block encodings }

\author[1,2]{Yuxin Zhang\thanks{zhangyuxin@amss.ac.cn}}
\author[2]{Changpeng Shao\thanks{changpeng.shao@amss.ac.cn}}

\affil[1]{{\small Alfr\'{e}d R\'{e}nyi Institute of Mathematics, Budapest, Hungary}}

\affil[2]{{\small SKLMS, Academy of Mathematics and Systems Science, Chinese Academy of Sciences, Beijing, China}}

\date{\today}

\begin{document}

\maketitle

\begin{abstract}


Block encodings are a basic interface between quantum algorithms and linear algebra, underlying a broad range of quantum matrix algorithms. Standard constructions based on linear combinations of unitaries (LCU) achieve optimal circuit depth, but typically require logarithmically many ancilla qubits. This raises a fundamental resource question: to what extent can quantum workspace be reduced without increasing circuit depth?

We study this question from both algorithmic and lower bound perspectives. 
For 
a Hermitian decomposition 
$A=\sum_{j=1}^L \alpha_j H_j$, with \(\|H_j\|=1\) and $\alpha=\sum_j|\alpha_j|$, we give two coherent $\varepsilon$-approximate block-encoding constructions. 
The first uses 1 ancilla qubit and has circuit depth 
$\widetilde O (L(\alpha/\varepsilon)^{o(1)})$,
while the second uses $O(\log\log(\alpha/\varepsilon))$ ancillas and achieves circuit depth $\widetilde {O}(L)$.

We then investigate whether these tradeoffs are inherent. For a broad coherent simulation architecture based on Suzuki product formulas, we prove a quantitative ancilla-depth tradeoff. In the polynomial-resource regime and for a constant number of coherent simulation rounds, the tradeoff takes the form
$\log(1/\varepsilon) \leq O\big((\log Q)^2+2^a\log Q\big),$
where $Q$ is the circuit depth normalized by the number of Hamiltonian terms, and \(a\) is the number of ancilla qubits. 
In particular, within this architecture, achieving circuit depth with only polylogarithmic dependence on $1/\varepsilon$ requires more than a constant number of ancilla qubits.
In a different and more general repeated query model for LCU, we prove an ancilla-query tradeoff. For balanced coefficients $1/L$ and approximation error $\varepsilon=\eta/L$, with fixed $0<\eta<1$, the tradeoff takes the form $2^a=\Omega_\eta ( L^2/(T+L) )$,
where $T$ is the number of oracle queries. This tradeoff has several interesting implications: to approximately implement LCU, one must use $\Omega (\log L)$ ancillas when $T=O(L^{\alpha})$ for some $\alpha<2$.
Moreover, in the exact case, the logarithmic ancilla requirement $a\geq \log L$ cannot be overcome by increasing the number of queries. We also prove a general version of this tradeoff for arbitrary nonnegative coefficients.

Finally, as an application of our low-ancilla constructions, we obtain an optimal quantum algorithm for normalized trace estimation in the DQC1 model, together with a matching query lower bound. Together, these results establish quantitative space-depth and space-query tradeoffs in two natural circuit models, providing a starting point for the study of analogous tradeoffs in more general quantum circuit models.




\end{abstract}

\newpage

\tableofcontents

\section{Introduction}

Block encoding is a fundamental primitive in quantum algorithms~\cite{chakraborty2019power}, providing a general interface between linear-algebraic operators and their coherent quantum implementation. By embedding a possibly non-unitary operator into a larger unitary transformation, block encodings enable a broad range of quantum algorithmic techniques, including Hamiltonian simulation~\cite{low2019hamiltonian} and quantum singular value transformation (QSVT)~\cite{gilyen2019qsvt}. Consequently, the resources required to construct a block encoding directly affect the complexity of many higher-level quantum algorithms.

A central resource in such constructions is the number of ancilla qubits.
Suppose an operator is expressed as a linear combination of $L$ simple unitaries.
The standard linear-combination-of-unitaries (LCU) framework~\cite{berry2015simulating} gives an exact block encoding using $O(\log L)$ ancilla qubits and circuit depth linear in $L$.
Although this construction is highly useful and depth-efficient, its ancilla register is used to coherently select among the $L$ terms in the decomposition. This coherent selection typically requires multi-qubit controlled operations, which can be costly to implement and are particularly undesirable on near-term quantum devices. 
Reducing the ancilla overhead while retaining optimal circuit depth is important for simplifying the controlled operations in the resulting circuits. Beyond these implementation considerations, the $O(\log L)$ ancilla overhead is also theoretically significant: while it may seem modest from a query-complexity perspective, it represents a nontrivial cost in terms of coherent workspace.

This motivates a basic question:
\begin{quote}
    {\em How much quantum workspace is necessary for depth-efficient coherent block encoding?}
\end{quote}
More concretely, can one reduce the number of ancilla qubits while retaining (near-)optimal circuit depth, or is additional depth unavoidable when coherent workspace is restricted?

This question is particularly natural for block encodings. 
Recent work has shown that logarithmic ancilla space is unavoidable for exact block encoding in a broad circuit model \cite[Theorem A.1]{chakraborty2025quantum}. A similar lower bound was given in \cite[Theorem 10]{vasconcelos2025methods} for matrix products.
Approximation changes the situation substantially: one may hope to use much less coherent workspace and compensate through additional circuit depth or repeated queries. What has been missing is a quantitative understanding of this tradeoff.

In this work, we study this problem from both directions. We first show that approximate block encodings can indeed be constructed with dramatically reduced ancilla space. We then prove that, in natural coherent circuit models, reducing workspace necessarily incurs additional depth or query complexity. The resulting picture reveals a genuine space-time tradeoff for coherent block encodings in two natural models.


\subsection{Low-ancilla coherent block encodings}

In this paper, we study approximate block encodings of operators given by Hermitian decompositions, with a focus on reducing the number of ancillas while retaining (near-)optimal circuit depth. 

Previous work has established close connections between Hamiltonian simulation and block encoding, with constant ancilla overhead and logarithmic overhead in circuit depth; see, for example,~\cite[Section 6 and Theorem 9]{low2017hamiltonian} and~\cite[Corollary 71]{gilyen2019qsvt}. We refine the direction from Hamiltonian simulation to block encoding using GQSP~\cite{motlagh2024generalized}: our construction requires only a single ancilla qubit and applies under a weaker norm condition (see Theorem \ref{thm:arcsin-constru}). We then apply this reduction to operators given by Hermitian decompositions
\begin{equation}
\label{eq:hermitian-decomp}
    A=\sum_{j=1}^L \alpha_j H_j,
    \qquad \alpha_j> 0,\quad \|H_j\|=1,
\end{equation}
which naturally arise in a wide range of physical systems. 
This reduces the problem of constructing low-ancilla block encodings to that of implementing the corresponding Hamiltonian evolution. We obtain two constructions, distinguished by the Hamiltonian-simulation method used.

Let $\alpha=\sum_j \alpha_j.$
Using $2k$-th order product formulas, we obtain a single-ancilla block encoding with circuit depth
$\widetilde{O} ( 5^{k-1} L \left(\alpha/\eps\right)^{{1}/{2k}} )$,
see \cref{thm:be-with-trotter}. Optimizing over $k$ leads to a circuit depth of $\widetilde{O}(L \exp(c \sqrt{\log(\alpha/\varepsilon)}))$, where $c=\sqrt{2\ln 5}$.
Alternatively, using multiproduct formulas, we obtain a block encoding with $O(\log\log(\alpha/\eps))$
ancilla qubits and circuit depth $\widetilde{O}(L)$, see \cref{thm:MPF}.


These two constructions suggest a tradeoff between ancilla space and circuit depth: with a single ancilla qubit, the dependence on the precision is super-polylogarithmic, whereas allowing $O(\log\log(\alpha/\varepsilon))$ ancilla qubits reduces the circuit depth to $\widetilde{O}(L)$. This naturally raises the question of whether this behavior is inherent. In particular, can one achieve circuit depth with polylogarithmic dependence on $1/\varepsilon$ using only a constant number of ancilla qubits, or does achieving such depth necessarily require additional coherent workspace?


\subsection{Ancilla-depth tradeoff for coherent product-formula simulation}

Characterizing such a tradeoff for completely general quantum circuits appears difficult. We therefore begin with a broad coherent simulation architecture, motivated by the constructions above and based on Suzuki product formulas. Extending these tradeoff bounds beyond this architecture remains an important open problem. Although our first tradeoff result is formulated for Hamiltonian simulation, it is directly relevant to block encoding because the two primitives are closely connected with only constant ancilla overhead and logarithmic depth overhead.

More precisely, 
we consider the following circuit architecture: Let $a$ be the number of ancilla qubits, and define
\be \label{intro:simulation circuit}
W(t):=(V_s \otimes I) \mathcal{P}_s(t) \cdots (V_2 \otimes I) \mathcal{P}_2(t) (V_1 \otimes I) \mathcal{P}_1(t) (V_0 \otimes I),
\ee
where $V_0,V_1,\ldots, V_s$ act on the ancilla system, and for $\ell\in[s]$
\be \label{intro:eq of Pt}
\mathcal{P}_\ell(t) = \sum_{r=1}^R \ket{r} \bra{r} \otimes 
\Big[S_{2k}(t/sm_{\ell r}) \Big]^{m_{\ell r}}, 
\quad R=2^a.
\ee
The architecture is specified by the following parameters:
\begin{itemize}
    \item[(1)] $k\geq 1$: the integer specifying the Suzuki order $2k$;
    \item[(2)] $s\geq 1$: the number of coherent simulation rounds;
    \item[(3)] $a\geq 0$: the number of ancilla qubits;
    \item[(4)] $\{m_{\ell r}\in \mathbb{N}:\ell=1,\ldots,s; r=1,\ldots,R\}$: arbitrary positive integer refinement parameters specifying the scale of each Suzuki branch;
    \item[(5)] $\{V_\ell\in U(R): \ell=0,1,\ldots,s\}$:  arbitrary controller unitaries applied between consecutive simulation rounds.
\end{itemize}

Throughout this architecture, the controller unitaries $V_\ell$ and the refinement parameters $m_{\ell r}$ are chosen independently of the specific Hamiltonian terms $H_j$. They may depend on the global parameters of the problem, such as $L$, $k$, $s$, $\varepsilon$, but not on the particular choice of the $H_j$'s.

For the purpose of measuring circuit depth, we adopt a generous convention: a controlled term evolution $ \sum_{r=1}^{R}|r\rangle\langle r|\otimes e^{-i\theta_r H_j}$, with fixed \(j\) and branch-dependent evolution time \(\theta_r\),
is counted as a single elementary term-evolution slot. In an actual implementation, a general multiplexed evolution may incur additional overhead depending on the controller. Since we ignore this overhead, our depth measure can only underestimate the physical circuit depth, and hence this convention only strengthens the resulting lower bound.

\begin{rmk}
In \eqref{intro:eq of Pt}, for simplicity, we divide the total simulation time uniformly among the $s$ coherent rounds, so that each round simulates a time interval $t/s$. More generally, one may introduce round-dependent time fractions $\tau_\ell\geq 0$ satisfying $\sum_{\ell=1}^s \tau_\ell=1$, and replace $S_{2k}(t/sm_{\ell r})$ with $S_{2k}(\tau_\ell t/m_{\ell r})$. Our proof extends to this more general setting with only minor modifications and yields essentially the same tradeoff. We therefore restrict to the uniform choice $\tau_\ell=1/s$ throughout.
\end{rmk}

We measure the circuit depth by the total number $D$ of sequential Hamiltonian term evolution slots in the circuit $W(t)$, see~\eqref{intro:simulation circuit}. For a Hamiltonian decomposition
$A=\sum_{j=1}^{L}\alpha_jH_j,$
a single order-$2k$ Suzuki step\footnote{For background on Suzuki product formulas, see Section~\ref{subsubsec:trotter}.} contains $\Theta(5^{k-1}L)$ elementary term-evolution slots. Under the depth convention above, the total depth is therefore
$D 
=
\Theta (
5^{k-1} L
\sum_{\ell=1}^{s}\max_r m_{\ell r}
).$
It is convenient to factor out the baseline linear dependence on the number $L$ of Hamiltonian terms. We therefore define the normalized simulation depth
\be \label{intro: eq of Q}
Q:=\frac{D}{L}
=
\Theta\!\left(
5^{k-1}
\sum_{\ell=1}^{s}\max_r m_{\ell r}
\right).
\ee
Thus, $Q$ captures the additional complexity arising from Suzuki recursion, branch refinement, and coherent reuse.


Our goal is to approximate $e^{-i A t}$ for all $|t|\le 1$:
\be \label{intro eq:C(t)}
\|C(t) - \eta e^{-i A t}\| \leq \varepsilon,
\quad
C(t):= (\bra{0^a}\otimes I) W(t) (\ket{0^a}\otimes I).
\ee
Here $\eta = \bra{0^a} V_s \cdots V_1 V_0 \ket{0^a}$ is the ideal no-error controller amplitude. In particular, since $\mathcal{P}_\ell(0)=I$ for every $\ell$, we have $C(0)=\eta I$, so the approximation condition is automatically satisfied at $t=0$.

\begin{thm}[Ancilla-depth tradeoff for coherent Suzuki simulation]
\label{thm:tradeoff}

Fix an architecture of the form \eqref{intro:simulation circuit}. 
Let \(\varepsilon\in(0,1/2)\). Suppose that, for every Hermitian decomposition
$A = \sum_{j=1}^L \alpha_j H_j$ with $L\geq 2$,
$\sum_{j=1}^L \alpha_j=1 $ and  $H_j^2=I$,
the architecture satisfies \eqref{intro eq:C(t)} for all $|t|\leq 1$, with $\eta\neq0$. Then
\be \label{intro:general tradeoff}
\log \frac{1}{\varepsilon}
\leq
O\left((k+s2^a) \log(s2^aQ) + \log \frac{1}{|\eta|}\right).
\ee
If in addition $s2^a \leq Q^{c_0}$ and $|\eta|\geq Q^{-c_1}$ for some constants $c_0,c_1>0$,\footnote{These are mild polynomial-resource assumptions: $s2^a\le Q^{c_0}$ excludes superpolynomial controller resources, while $|\eta|\ge Q^{-c_1}$ excludes a superpolynomially small effective amplitude that would require substantial additional normalization or amplification. In short, they rule out resource regimes that are already inefficient for simulation.} then
\be \label{intro:simplified tradeoff}
\log \frac{1}{\varepsilon}
\leq
O\Big(  (\log Q)^2 +  s2^a \log Q \Big).
\ee
\end{thm}

To see the above tradeoff more clearly,
we write $x=\log(1/\varepsilon), y = \log Q$, then for constant $s$,\footnote{The theorem itself allows arbitrary $s$. We focus here on the natural regime $s=O(1)$ to isolate the tradeoff between ancilla space and simulation depth. Allowing $s$ to grow provides an additional computational resource, whose effect is already captured by the factor $s2^a$ in the general bound.} \eqref{intro:simplified tradeoff} becomes $x \leq O(y^2 + 2^a y)$. 
We first consider the low-ancilla regime. If $2^a\leq \sqrt{x}$, which includes the constant-ancilla case $a=O(1)$, then $y=\Omega(\sqrt{x})$. This gives
\be
Q\geq \exp\left(\Omega\left(\sqrt{\log(1/\varepsilon)}\right)\right).
\ee
This recovers the familiar precision dependence of optimized Suzuki product-formula methods.
In the high-ancilla regime, if $2^a\geq \sqrt{x}$, then
$y=\Omega ( x/2^a),$
or equivalently,
\be
Q\geq \exp\left(\Omega\left(
\frac{\log(1/\varepsilon)}{2^a}
\right)\right).
\ee
Thus increasing the coherent workspace weakens the depth obstruction.
In particular, when $a=\log\log(1/\varepsilon)+O(1)$, we have
$2^a=\Theta(\log(1/\varepsilon))$, and the lower bound above becomes constant. 
It therefore no longer rules out polylogarithmic depth, consistent with the regime achieved by multiproduct formula methods.

Interestingly, the tradeoff \eqref{intro:simplified tradeoff} bears a formal resemblance to the small-error quantum search lower bound of Buhrman, Cleve, de Wolf, and Zalka \cite[Theorem 4]{BCWZ99}, which takes the form
$
\log(1/\varepsilon)
=O\bigl(T^2/N+T\sqrt{t/N}\bigr),
$
where $T$ is the query complexity, $N$ is the search-space size, and $t$ is a promised lower bound on the number of marked items. Although these bounds arise in different computational settings and are proved using different techniques, their similar quadratic-linear structure suggests the possibility of a common underlying approximation-theoretic principle.

\begin{rmk}
We make several remarks on Theorem~\ref{thm:tradeoff}.

\begin{itemize}
\item The condition $H_j^2=I$ restricts the Hamiltonian simulation problem
to Hermitian involutions, but this is a natural subclass from the block-encoding perspective.
For a decomposition $A=\sum_j \alpha_j H_j$, every Hermitian contraction $H_j$ can be embedded into a Hermitian unitary, for example via the standard dilation
\[
\widetilde H_j=
\begin{bmatrix}
H_j & \sqrt{I-H_j^2}\\
\sqrt{I-H_j^2} & -H_j
\end{bmatrix},
\qquad
\widetilde H_j^2=I.
\]
Moreover, 
as discussed above,
Hamiltonian simulation and block encoding can be converted
into one another with only constant ancilla overhead and logarithmic
overhead in the target precision. Thus, a lower bound already for Hermitian involutions captures an obstruction directly relevant to coherent block encodings.



\item The general bound \eqref{intro:general tradeoff} holds without further restrictions on the Suzuki order $2k$, the number of coherent rounds $s$, the ancilla size $a$, the refinement parameters $\{m_{\ell r}\}$, or the controller unitaries $\{V_\ell\}$.

\item  The general tradeoff \eqref{intro:general tradeoff} has a simple interpretation in terms of two distinct sources of approximation power.
Increasing the Suzuki order $2k$ gives higher-order cancellation, but this also increases the circuit depth; since $Q\gtrsim 5^{k-1}$, one always has $k=O(\log Q)$. Thus, increasing the Suzuki order alone can contribute at most $O (\log Q\cdot \log(s2^aQ) )$.
Additional error cancellation comes from coherent control and is governed by $s2^a$. 
Hence, if the number of ancilla qubits is reduced, maintaining the same cancellation power requires increasing $s$, which in turn increases the simulation depth. 

\end{itemize}
\end{rmk}

The tradeoff above is derived from the specific structure of Suzuki-based simulation. To understand whether a similar limitation persists beyond product-formula architectures, we next turn to the closely related LCU setting, where the constituent operators are unitaries.

\subsection{Ancilla-query tradeoffs for approximate LCU}

Hamiltonian decompositions are closely related to LCU decompositions, since Hermitian terms can often be embedded into Hermitian unitaries by a standard dilation. This raises a basic question: do analogous resource tradeoffs persist for coherent LCU implementations?

We study this question in a general repeated-query circuit model and establish a quantitative tradeoff between ancilla space and query complexity. In the approximate setting, additional queries can compensate for reduced ancilla space, but only to a limited extent. In the exact case $\varepsilon=0$, our bound shows that the logarithmic ancilla requirement persists regardless of the number of repeated queries, extending the exact lower-bound picture of~\cite{chakraborty2025quantum}.

We now state the problem and our main result more precisely.
We consider the task of approximately implementing an LCU using ancillary qubits. Specifically, given $P_1,\ldots,P_L\in U(N)$ and coefficients $\lambda_1,\ldots,\lambda_L\in \mathbb{R}_{\geq 0}$ satisfying \(\sum_{j=1}^{L}\lambda_j=1\), our goal is to approximate $\sum_{j=1}^L \lambda_j P_j$ to additive error $\varepsilon$ in operator norm. We consider the following repeated-query circuit framework:
\be\label{intro:circuit for LCU}
U = (V_T\otimes I) C(P_{\alpha_T}) \cdots (V_2\otimes I) C(P_{\alpha_2}) (V_1\otimes I) C(P_{\alpha_1})(V_0\otimes I) ,
\ee
where
\begin{itemize}
    \item[(1)] $V_0,V_1,\ldots,V_T\in U(2^a)$ are unitaries acting on the ancilla space,  $a$ is the number of ancillas;
    \item[(2)] $\alpha_1,\alpha_2,\ldots,\alpha_T\in \{1,2,\ldots,L\}$ are query indices and may be repeated;
    \item[(3)] $C(P_{\alpha_j})=\ket{0^a}\bra{0^a}\otimes P_{\alpha_j}+(I-\ket{0^a}\bra{0^a})\otimes I$ is the $0^a$-controlled $P_{\alpha_j}$ operator.\footnote{Allowing control on a different fixed ancilla basis state does not change the model, since such controls can be converted to $\ket{0^a}$-controls by conjugating with suitable ancilla unitaries.}
    \item[(4)] $T$ denotes the total number of queries to the input unitaries \(P_1,\ldots,P_L\), or equivalently, the number of controlled-\(P_j\) gates appearing in the circuit. Under the convention that the cost of unitaries $V_\ell$ and the internal implementation of each \(P_j\) are not counted, \(T\) is also the circuit depth of this query model.
\end{itemize}
The goal is
\be
\label{intro:circuit goal}
\Big\| B - \sum_{j=1}^L \lambda_j P_j \Big\| \leq \varepsilon, \quad B = (\bra{0^a}\otimes I) U(\ket{0^a}\otimes I) .
\ee


For clarity, we only state the balanced case, which yields the cleanest form of the tradeoff. In Section \ref{sec:proof of thm2}, we establish a more general ancilla-query tradeoff for arbitrary nonnegative coefficients.

\begin{thm}[Ancilla-query tradeoffs for approximate LCU]
\label{thm:balancedcase}
Assume \eqref{intro:circuit goal} holds for all unitaries $P_1,\ldots,P_L\in U(N)$ with $\lambda_j=1/L$ for all $j$, and $\varepsilon=\eta/L$ for some fixed $\eta\in[0,1)$. 
Assume $N\geq T(L+1)^2$.
Then\footnote{The condition $N\geq T(L+1)^2$ is typically satisfied, since $T$ and $L$ are usually of order $\polylog(N)$.
It comes from our proof technique rather than from the LCU problem itself, and we do not expect it to be intrinsic to the ancilla-query tradeoff.
Here $\Omega_\eta$ means the constants only depend on $\eta$, and $\eta=L \varepsilon$ measures the approximation error relative to the  coefficient scale \(1/L\).}
\be \label{intro:thm-eq}
2^a \geq \frac{L^2}{L+2K_\eta^2 (T-1)}=\Omega_\eta\left(\frac{L^2}{T+L} \right),
\ee
where 
\be \label{intro:eq:K-constant}
K_\eta = \frac{\eta}{1-\sqrt{\eta(2-\eta)}}+\sqrt{\frac{2\eta}{1-\eta}}.
\ee
Particularly, we have the following special cases:
\begin{itemize}
    \item If $T=O(L^{1+\alpha})$ for some $0\le \alpha<1$, then $2^a=\Omega_\eta(L^{1-\alpha})$, namely, $a\geq (1-\alpha)\log(L)-O_\eta(1)$.
    \item If $2^a=O(L^\beta)$ for some $0\le\beta<1$, then $T=\Omega_\eta(L^{2-\beta})$.
    \item If $\eta=\varepsilon=0$, then $a\geq \log(L)$, for arbitrary $T$.
\end{itemize}
\end{thm}

The theorem reveals a qualitative distinction between the exact and approximate settings. In the approximate case, reduced ancilla space can in principle be compensated for by additional queries, although this tradeoff is quantitatively constrained. In contrast, when $\varepsilon=0$, we have $2^a\ge L$ independently of $T$. Thus, no ancilla-query tradeoff is possible for exact implementation: additional queries cannot compensate for reduced ancilla space. In particular, the standard LCU construction is optimal in its ancilla scaling within the circuit model~\eqref{intro:circuit for LCU}.

As $\eta\to0$, we have $K_\eta=\sqrt{2\eta}+O(\eta)$. Hence, in the balanced case,
$2^a
\gtrsim
\frac{L^2}{L+4\eta T}.$
In particular, when $T\ll L/\eta$, the bound remains $2^a=\Omega(L)$, and the dependence on $T$ becomes significant only when $T$ reaches order $L/\eta$.

\subsection{Relation between the two tradeoffs}

The two tradeoffs concern different circuit models and exhibit qualitatively different low-ancilla behavior. The Suzuki result concerns coherent Hamiltonian simulation and gives a precision-sensitive ancilla-depth tradeoff within a broad product formula architecture, directly connecting to our low-ancilla constructions. In contrast, 
the repeated-query LCU result treats the constituent operators as arbitrary unitary black boxes, accessed through repeated controlled queries interleaved with ancilla unitaries.

This distinction becomes particularly clear in the constant-ancilla regime. Consider the balanced case $\lambda_j=1/L$ and the precision scale $\varepsilon=\eta/L$ for a fixed constant $\eta\in(0,1)$. The repeated-query LCU lower bound gives
$
T=\Omega_\eta(L^2)
$
for $a=O(1)$. By contrast, the Hamiltonian-simulation-based construction in \cref{thm:be-with-trotter} achieves
$
\widetilde{O} (L/\varepsilon^{o(1)})
=
\widetilde{O}(L^{1+o(1)})
$
depth with a single ancilla at the same precision scale.
Thus, the quadratic obstruction should be viewed as a structural limitation of the generic repeated-query LCU architecture, rather than as a general limitation of low-ancilla block encoding. The comparison indicates that the additional algebraic structure available in the Hamiltonian-evolution setting can fundamentally change the workspace–complexity tradeoff. At the same time, neither lower bound applies to completely unrestricted quantum circuits, and extending such tradeoffs to more general circuit models remains an important open problem.

\subsection{Technical overview of the tradeoffs}

We first briefly explain the main idea behind the ancilla-depth lower bound in Theorem \ref{thm:tradeoff}. The key idea is that higher-order Suzuki formulas intrinsically cancel low-order errors, while ancilla-assisted coherent combinations can cancel additional error terms. Since the available degrees of freedom for cancellation are limited, the error coefficients cannot vanish to arbitrarily high order.

More precisely, we study the linear response of the simulation circuit around a carefully chosen one-parameter family of Hamiltonians. In particular, we choose
$
H_1(\delta)=-(\cos\delta)Z+(\sin\delta)X$ and
$H_j=Z\; (j\geq 2),
$
with coefficients such that the corresponding Hamiltonian $A(\delta)$ satisfies $A(0)=0$ and $A'(0)=X/2$.
Thus, at $\delta=0$, the target evolution is simply the identity, while its first-order response is nontrivial:
$
\partial_\delta e^{-itA(\delta)}|_{\delta=0}
=-itX/2.
$
Writing the linear response of the simulator as
$\left.\partial_\delta C(t,\delta)\right|_{\delta=0}
=
-\frac{i\eta t}{2}X+g(t)X,$
the function $g(t)$ measures the mismatch between the simulator and the target at first order in $\delta$.
Therefore, a high-accuracy simulation of this family must also accurately reproduce this linear response.

For an order-$2k$ Suzuki formula, the linear response already agrees with the exact evolution up to an order determined by $k$. The coherent controller can further cancel higher-order error terms by combining different Suzuki branches. With $a$ ancilla qubits, the controller space has dimension $2^a$, and over $s$ coherent reuse rounds, there are at most $O(s2^a)$ branches available for such cancellations. After expanding $g(t)$ in powers of $t$, the additional cancellation conditions take the form
$
M_n=\sum_{j=1}^{K} c_j m_j^{-2n}=0,
$
for some $K\leq s2^a$, where each $m_j$ is a refinement parameter associated with the different branches.
To obtain a quantitative lower bound, we show that this surviving coefficient cannot be arbitrarily small. More precisely, for some $n\leq k+s2^a-1$, we have
$
|g^{(2n+1)}(0)|
\geq|\eta|\exp (-O ((k+s2^a)\log(s2^aQ))).
$
On the other hand, if the whole simulation approximates the target evolution within error $\varepsilon$, then the corresponding linear-response error must be uniformly small. 
A trigonometric Bernstein inequality converts this error into a bound on \(\|g\|_\infty\), while a Markov-type inequality for exponential sums bounds the derivatives of \(g\) in terms of this uniform norm. Combining these two estimates gives the tradeoff \eqref{intro:general tradeoff}.

We now illustrate the main idea behind Theorem \ref{thm:balancedcase}. We first expand the repeated-query circuit as a noncommutative polynomial in the input unitaries $P_1,\ldots,P_L$. Since the target is linear in these unitaries, the degree-one component must reproduce all $L$ terms with comparable weight, while all higher-degree components must remain small.

This translates into constraints on the circuit's query histories. The degree-one condition requires many distinguishable histories, whereas the suppression of higher-degree terms forces suitable shifted histories to have small overlaps. We encode these overlaps in a Gram matrix and compare two rank bounds: a lower bound coming from the large family of nearly independent query histories, and an upper bound coming from a shift-dilation representation whose auxiliary dimension is controlled by the ancilla dimension $2^a$. Comparing the two yields the ancilla-query tradeoff \eqref{intro:thm-eq}.

\subsection{An optimal quantum algorithm for normalized trace estimation in DQC1}

Finally, we apply our low-ancilla block-encoding construction to normalized trace estimation in DQC1. Given a continuous function $f$ and an $N=2^n$ dimensional log-local Hamiltonian $A$, the goal is to estimate $\frac{1}{N}\tr[f(A)].$
This problem was studied by Cade and Montanaro~\cite{cade2018quantum} using quantum phase estimation, with complexity linear in the Lipschitz constant of the function being implemented. If a degree-$d$ polynomial approximation to $f$ is used, Markov brothers' inequality~\cite{Markov} gives a Lipschitz bound of order $O(d^2)$, leading to a quadratic dependence on the polynomial degree.

Using our few-ancilla block encoding of $A$ together with QSVT, we instead obtain an algorithm whose complexity is linear in the approximate degree of $f$; see Theorem~\ref{thm:trace-qsvt}. We further prove a matching linear lower bound in the corresponding query model, under the assumptions of Theorem~\ref{thm:approx-degree-lb}. Thus, within this setting, the dependence on the approximate degree is optimal.

\subsection{Related work}

\paragraph{Low-ancilla constructions.}
Reducing the ancilla overhead of block-encoding-based quantum algorithms has attracted increasing attention in recent years. 
The work most closely related to ours is \cite{chakraborty2025quantum}, in which Chakraborty et al. developed a framework for performing QSVT-type transformations without explicitly constructing block encodings. Their approach uses only a single ancilla qubit and achieves near-optimal complexity by combining Hamiltonian simulation with Richardson extrapolation. Unlike the standard QSVT framework, however, the resulting transformation is implemented at the level of expectation-value estimation rather than as a coherent block encoding.
Kang and Su \cite{KangSu26} developed a Hamiltonian-block-encoding framework in which matrix arithmetic is performed directly through Hamiltonian evolution, allowing a broad class of matrix operations to be composed using at most two ancilla qubits. In contrast to our setting, their framework maintains a Hamiltonian representation of the matrix rather than constructing a standard coherent block encoding; the two approaches nevertheless share the idea of exploiting Hamiltonian evolution to reduce ancilla overhead.
More directly related to block encoding, Vasconcelos and Gily\'{e}n \cite{vasconcelos2025methods} developed general methods for reducing ancilla overhead for approximate multiplication of block encodings. 
Another line of work reduces quantum-space overhead through randomization, such as qDRIFT
\cite{Campbell2019qDRIFT}, qubit-efficient randomized algorithms
\cite{WangMcArdleBerta2024}, single-ancilla LCU
\cite{Chakraborty2024LCU}, and randomized QSVT
\cite{WangEtAl2025RandomizedQSVT}, etc.

\paragraph{Ancilla tradeoffs and lower bounds.}
Time-space tradeoffs have a long history in quantum algorithms and complexity. For example, Klauck, Špalek, and de Wolf established optimal time-space tradeoffs for quantum sorting and related problems \cite{KlauckSpalekDeWolf}, while more recent works of Jeffery and collaborators studied space-efficient quantum algorithms and time-space tradeoffs for graph connectivity problems \cite{ApersJefferyPassWalter, JefferyPass}.

In the context of block encodings or Hamiltonian simulation, Chakraborty et al. \cite{chakraborty2025quantum} proved that $\Omega(\log L)$ ancilla qubits are necessary for exact block encoding of a linear combination of $L$ unitaries in a fairly general circuit model.
Our result strengthens this lower bound in two respects. First, in the exact setting, we show that the $\Omega(\log L)$ ancilla lower bound continues to hold even when the circuit is allowed an arbitrary number of repeated queries to the input unitaries. Second, we extend the lower bound to approximate block encodings and establish a quantitative ancilla-query tradeoff, showing explicitly how additional queries can compensate for reduced ancilla space.
Vasconcelos and Gily\'{e}n \cite{vasconcelos2025methods} established complementary space-time and space-accuracy tradeoffs for ancilla reduction and multiplication of block encodings.
Ma, Joven, and Liu \cite{ma2026optimal} studied the interplay between ancilla space and non-Clifford resources for block encodings, including an ancilla-compression theorem in the Clifford$+T$ model. 
Zlokapa, Allen, and Harrow \cite{zlokapa2026optimal} established tight gate and query lower bounds for Hamiltonian simulation, showing in particular that the complexity suggested by coherent oracle queries in block-encoding models can differ substantially from physical gate complexity.  In contrast, our results explicitly quantify how ancilla space can be traded against circuit depth or query complexity in coherent block-encoding and LCU architectures.
Recent work by Wada et al. \cite{Wada2025} explored quantum–classical resource tradeoffs for LCU by combining coherent implementations of smaller LCU blocks with randomized sampling. In contrast, our work focuses on the ancilla–depth tradeoff for fully coherent block-encoding constructions.

\subsection*{Organization of the paper}

The paper is organized as follows. Section \ref{sec:BE-via-HS} reviews preliminary results and gives a single-ancilla construction of block encodings from Hamiltonian evolution. Section \ref{sec:Low-ancilla block encodings for Hermitian decompositions} develops two low-ancilla block-encoding algorithms for operators given by Hermitian decompositions. Sections \ref{sec:proof of thm 1} and \ref{sec:proof of thm2} prove Theorems~\ref{thm:tradeoff} and \ref{thm:balancedcase}, respectively. Finally, Section \ref{sec:applications} presents the application to normalized trace estimation in the DQC1 model.

\section{From Hamiltonian simulation to block encoding}
\label{sec:BE-via-HS}

In this section, we show how to construct block encodings from Hamiltonian simulation via generalized quantum signal processing (GQSP). We first recall the definition of block encoding and the standard LCU construction, which serves as a benchmark for ancilla overhead. We then review the GQSP result for bounded Laurent polynomials. Finally, we construct a desired Laurent polynomial which, given access to controlled Hamiltonian evolution operators, yields a block encoding of the Hamiltonian itself using only $1$ ancilla qubit.

\subsection{Block encoding and LCU}
\label{subsec:be-lcu}

Block encodings provide a standard way to represent a possibly non-unitary operator as part of a larger unitary operator, enabling one to apply quantum algorithms, such as QSP and QSVT, to general linear operators. The notion was implicit in earlier Hamiltonian simulation algorithms~\cite{low2019hamiltonian} and was later formalized in~\cite{chakraborty2019power,gilyen2019qsvt}. We use the following standard definition.

\begin{defn}[Block encoding]
   Let $A$ be an $s$-qubit operator, $\alpha,\eps\in\mb R_{\geq 0}$ and $a\in \mb N$, then the $(s+a)$-qubit unitary $U_A$ is called an $(\alpha, a, \varepsilon)$-block encoding of $A$, if $\left\|A-\alpha\cdot(\bra{0^a} \otimes I) U_A (\ket{0^a} \otimes I) \right\| \leq \varepsilon$.
\end{defn}

For a unitary decomposition $A=\sum_{j=1}^L \alpha_j U_j$ with $\alpha_j>0$ and $\alpha=\sum_j \alpha_j$, the standard LCU~\cite{ChildsWiebe2012,berry2015simulating} framework yields an $(\alpha, \lceil \log L \rceil, 0)$-block encoding of $A$ as follows.
Let $V$ be a unitary such that
$V\ket{0}=\sum_{j=1}^L \sqrt{\alpha_j/\alpha} \ket{j}$,
then the unitary
$(V^\dag \otimes I)(\sum_{j=1}^{L} \ket{j}\bra{j} \otimes U_j) (V\otimes I)$
encodes exactly $A/\alpha$ in its top-left block, with circuit depth $O(L)$.

In~\cite{chakraborty2025quantum}, it was proved that this $O(\log L)$ ancilla overhead is unavoidable for exact block-encoding constructions,
for a fairly general circuit model, including LCU as a special case.
\begin{lem}[Theorem~A.1 of \cite{chakraborty2025quantum}]
\label{lem:extact LCU lower bound}
    For the circuit model
    \[
    U = (V_1 \otimes I) c_0\text{-}U_1 (V_2 \otimes I) c_0\text{-}U_2 \cdots 
    (V_L \otimes I) c_0\text{-}U_L (V_{L+1} \otimes I),
    \]
where $V_1,\ldots,V_{L+1}$ act on $a$ ancilla qubits, and each $c_0\text{-}U_j$ is controlled on $\ket{0^a}$. If $U$ is an exact block encoding of $A$ for every choice of unitaries \(U_1,\ldots,U_L\), then $a=\Omega(\log L)$.
\end{lem}

Motivated by this ancilla lower bound for exact LCU constructions, we turn to the more general setting in which only a Hermitian decomposition is given. In this case, the standard LCU construction is not directly applicable.

\begin{restatable}{prob}{additiveBEprob}
\label{prob:additive-be}
Assume that $A=\sum_{j=1}^L \alpha_j H_j$ is a Hermitian decomposition, where $\alpha_j\in\mathbb{R}_{\geq 0}$ and $\|H_j\|=1$. How efficiently can one construct an $\eps$-approximate block encoding using only 1 ancilla?
\end{restatable}


\subsection{Generalized quantum signal processing}

We recall the GQSP result used in our construction, and refer the reader to~\cite{motlagh2024generalized} for further details.
Suppose we can implement the controlled Hamiltonian simulation $U=e^{i H}$ (and  $U^{\dagger}=e^{-i H}$) as a black box. That is, the operators 
\[
    c_0\text{-}U
    =\left[\begin{array}{cc}
        U & 0 \\
        0 & I
        \end{array}\right], \quad 
    c_1\text{-}U^\dag 
    =\left[\begin{array}{cc}
            I & 0 \\
            0 & U^{\dagger}
            \end{array}\right]
\]
using one query to $U$ or $U^{\dagger}$. Let 
\begin{align*}
    R(\theta, \phi, \lambda)=\left[\begin{array}{cc}
        e^{i  (\lambda+\phi)} \cos (\theta) & e^{i   \phi} \sin (\theta) \\
        e^{i   \lambda} \sin (\theta) & -\cos (\theta)
        \end{array}\right] \otimes I 
\end{align*}
be an arbitrary $U(2)$ rotation of the single ancilla qubit.  
Motlagh and Wiebe \cite{motlagh2024generalized} showed that for any degree-$n$ polynomial $P(z)$, satisfying $|P(z)|\le 1$ on $\mathbb{T}:= \{ z\in \mathbb{C} : |z|=1 \}$, there exists an interleaved sequence of $R(\theta_j, \phi_j, 0)$ and $c_0\text{-}U$, $c_1\text{-}U^\dag$, of length $2n+1$, that implements a block encoding of $P(U)$.  
Moreover, their framework also holds for Laurent polynomials bounded on $\mathbb{T}$. We state this result as follows:

\begin{lem}[Combining Corollary 5 and Theorem 6 of \cite{motlagh2024generalized}]
    \label{lem:generalized-QSP}
    For any Laurent polynomial $P(z) = \sum_{j=-n}^{n} a_{j} z^j$ such that $|P(z)| \le 1$ for all $z\in \mathbb{T}$, there exist $\Theta=(\theta_j)_j, \Phi=(\phi_j)_j \in \mathbb{R}^{2n+1}, \lambda \in \mathbb{R}$ such that 
    \begin{equation*}
    \label{eq:gqsp}
        \begin{bmatrix}
            P(U) & * \\
            * & * 
        \end{bmatrix} = \Big(\prod_{j=1}^n R(\theta_{n+j}, \phi_{n+j}, 0) c_1\text{-}U^\dag \Big)\Big(\prod_{j=1}^{n} R(\theta_j, \phi_j, 0) c_0\text{-}U \Big) R(\theta_0, \phi_0, \lambda).
    \end{equation*}
\end{lem}

\subsection{Block encoding using Hamiltonian simulation}
\label{subsec:BE-via-HS}
It is well known that, given a block encoding of a normalized Hamiltonian $H$ with circuit depth $d$, we can use QSP to construct an $\varepsilon$-approximate block encoding of $e^{iH}$ with $2$ additional ancilla qubits and circuit depth $\Theta(d\log(1/\varepsilon))$~\cite{low2017optimal}.
In this section, we show how to construct block encodings by combining generalized quantum signal processing (GQSP)~\cite{motlagh2024generalized} with Hamiltonian simulation.
The same idea was previously used in the QSVT framework, applying the arcsine function to a block encoding of $\sin(H)$ obtained from Hamiltonian evolution; see~\cite[Corollary~71]{gilyen2019qsvt} and \cite[Theorem 9]{low2017hamiltonian}.
Here we instead use GQSP directly on the Hamiltonian evolution operator, yielding a simple single-ancilla construction.

For any Hamiltonian $H$, GQSP implements a degree-$n$ bounded Laurent polynomial $P(e^{iH})$, using $O(n)$ queries to the controlled evolution operators $e^{\pm iH}$ and only $1$ ancilla. 
To construct a block encoding of $H$ using GQSP, it therefore suffices to construct a bounded Laurent polynomial $P$ such that $P(e^{ix})\approx x$ on an appropriate interval.
We present a method to construct such a Laurent polynomial, based on \cite[Lemma 70]{gilyen2019qsvt} and the identity $x=\arcsin(\sin x)$ on $[-1,1]$.

\begin{lem}[Lemma 70 of \cite{gilyen2019qsvt}]
\label{lem:approximation of arcsin}
    Let $\delta, \varepsilon \in\left(0, \frac{1}{2}\right]$, then there is an odd polynomial $P \in \mathbb{R}[x]$ of degree $O\left(\frac{1}{\delta} \log \left(\frac{1}{\varepsilon}\right)\right)$ such that $\|P\|_{[-1,1]} \leq 1$ and
    \[
    \left\|P(x)-\frac{2}{\pi} \arcsin (x)\right\|_{[-1+\delta, 1-\delta]} \leq \varepsilon.
    \]
\end{lem}

\begin{thm}
\label{thm:arcsin-constru}
    Let $A$ be a Hermitian operator with $\|A\|\leq 1$. Suppose that controlled-$e^{\pm iA}$ can be implemented with circuit depth $d$. Then there is a quantum circuit of depth $O(d\log(1/\eps))$ that implements a $(\pi/2, 1, \eps)$-block encoding of $A$.
\end{thm}

\begin{proof}
We have the following identity $\sin(A) = \frac{e^{iA} - e^{-iA}}{2i}$.
By~\cref{lem:approximation of arcsin} with $\delta=1-\sin(1)$, there exists an odd polynomial $P\in\mb R[x]$ of degree $n=O(\log(1/\tilde{\eps}))$ such that
\[
\left\|P(x)-\frac{2}{\pi}\arcsin(x)\right\|_{[-\sin(1), \sin(1)]}\leq \tilde\varepsilon
\qquad\text{and}\qquad
\|P\|_{[-1,1]}\leq 1.
\]
Define the Laurent polynomial
$Q(z):=P (\frac{z-z^{-1}}{2i} )$.
Then $Q$ has degree at most $n$ and for any $z=e^{i\theta}\in\mb T$, we have $|Q(z)|=|P(\sin \theta)| \leq 1$.
Also, we have
\[
    Q(e^{i A})=P(\sin A)\approx_{\tilde\eps} \frac{2}{\pi} \arcsin (\sin A) =  \frac{2}{\pi} A
\]
for $\|A\|\leq 1$.
Applying GQSP, we obtain a $(1,1,\tilde\eps)$-block encoding of $2A/\pi$. Let $\tilde\eps=2\eps/\pi$, then we obtain a $(\pi/2,1,\eps)$-block encoding of $A$. The circuit uses $O(n)$ queries to controlled-$e^{\pm iA}$, and therefore has circuit depth 
$O(nd)=O(d\log(1/\eps)).$
\end{proof}


\section{Low-ancilla block encodings for Hermitian decompositions}
\label{sec:Low-ancilla block encodings for Hermitian decompositions}


In this section, we apply \cref{thm:arcsin-constru}
to \cref{prob:additive-be}.
Instead of assuming access to Hamiltonian evolution operators as in \cref{thm:arcsin-constru}, we implement them approximately from the given decomposition. 
We consider two approaches. The first uses higher-order Trotterization and yields a strictly single-ancilla block encoding, at the cost of $(1/\varepsilon)^{o(1)}$ dependence in the circuit depth. The second uses multiproduct formulas and improves this precision dependence to polylogarithmic, while increasing the ancilla overhead to $O(\log\log(1/\varepsilon))$.

\subsection{Single-ancilla block encodings via product formulas}

Hamiltonian simulation is one of the primary applications of quantum computing~\cite{lloyd1996universal}. Trotterization, also known as product formulas, provides a constructive tool for implementing the Hamiltonian evolution operator $e^{-iAt}$ by decomposing $A$ into simple Hermitian terms $H_j$ and simulating each $e^{-iH_jt}$.
Over the years, higher-order constructions have been developed, yielding efficient Hamiltonian simulation algorithms while retaining simple implementation \cite{berry2007efficient,childs2018toward,childs2021theory}.

\subsubsection{Product formulas}
\label{subsubsec:trotter}
We recall some basic results on product formulas.
For a Hamiltonian expressed as a sum of $L$ Hermitian terms $A = \sum_{\gamma=1}^{L} H_{\gamma}$, a \textit{staged product formula} is
\[
\mathcal{P}(t):= \prod_{\nu=1}^{\Upsilon} \prod_{\gamma=1}^{L} e^{ i t a_{(\nu, \gamma)} H_{\pi_\nu(\gamma)}},
\]
where $\pi_\nu$ are permutations and $a_{(\nu, \gamma)} \in \mb R$. A product formula $\mathcal{P}(t)$ is said to be of {\em order $p$} if $e^{-i A t} = \mathcal{P}(t)+O(t^{p+1})$, and {\em symmetric} if $\mathcal{P}(-t)=\mathcal{P}^{-1}(t)$. 

A well-known example is the $2k$-th order Trotter-Suzuki formula $S_{2k}(t)$, recursively defined as
\be
\label{eq:2nd-trotter}
S_2(t):= \prod_{\gamma=L}^1 e^{-i H_\gamma t / 2} \prod_{\gamma=1}^{L} e^{-i H_\gamma t / 2}
\ee
and
\be
\label{eq:trotter-suzuki}
S_{2 k}(t):= [S_{2(k-1)}(u_k t)]^2 S_{2(k-1)}((1-4 u_k) t)[S_{2(k-1)}(u_k t)]^2,
\ee
where $u_k=1/(4-4^{1/(2k-1)})$.
It is known \cite{berry2007efficient,wiebe2010higher} that $S_{2k}(t)$ is symmetric with $\Upsilon=2\cdot 5^{k-1}$ and $|a_{(\nu, \gamma)}|\le 2k/3^k$.

\begin{lem}[Theorem 6 and Corollary 7 of \cite{childs2021theory}]
\label{lem:Trotter-method}
Let $A=\sum_{j=1}^L H_j$ be a Hermitian decomposition and $t\geq 0$. Let $S(t)$ be a $p$th-order $\Upsilon$-stage product formula and 
\be
\label{eq:lambda}
{\lambda}_{\rm comm} = \sum_{j_1,\ldots, j_{p+1}=1}^L \big\| [H_{j_{1}},[H_{j_{2}},[ \cdots ,[H_{j_p},H_{j_{p+1}}] \cdots]]] \big\|.
\ee
Then
$\|S(t) - e^{-i A t}\| = O({\lambda}_{\rm comm} t^{p+1}).$
Moreover, $\|S^r(t/r) -  e^{-iAt}\| \leq \varepsilon$ provided that
\[
r = O\left( \frac{{\lambda}_{\rm comm}^{1/p} t^{1+1/p} }{\varepsilon^{1/p}} \right) .
\]
\end{lem}

\begin{rmk}
A trivial upper bound for ${\lambda}_{\rm comm}$ is
${\lambda}_{\rm comm} \leq 2^p  (\sum_{j=1}^L \|H_j\| )^{p+1}.$
\end{rmk}

\subsubsection{Block encoding using Trotterization}
\label{subsec:be-with-trotter}
Combining the construction based on GQSP in \cref{thm:arcsin-constru} with higher-order Trotterization described in \cref{subsubsec:trotter}, we obtain one way to solve \cref{prob:additive-be}.

\begin{thm}
\label{thm:be-with-trotter}
    Assume that $A=\sum_{j=1}^L \alpha_j H_j$ is a Hermitian decomposition with $\alpha_j \in \mathbb{R}_{\geq 0}$ and $\|H_j\|=1$. Let $\alpha=\sum_j \alpha_j$ and $k \in \mathbb{N}$.
    Then there is a quantum circuit that implements a $(\alpha\pi/2, 1, \eps)$-block encoding of $A$, with the circuit depth
    $\widetilde{O}( 5^{k-1} L (\alpha/\eps)^{\frac{1}{2k}})$
    assuming Hamiltonian simulation of each $H_j$ costs $O(1)$.
\end{thm}
\begin{proof}
    Let $\tilde{A}=A/\alpha$ and $U=e^{i\tilde{A}}$. We use the $2k$-th order Trotter-Suzuki formula $S(t)$ to approximate $U$. By \cref{lem:Trotter-method}, we have $\|S^r(-1/r)-e^{i\tilde{A}}\|=O(\eta)$ when
    $r=O ( (\tilde\lambda_{\rm comm}/\eta )^{\frac{1}{2k}}  ),$
    where $\tilde\lambda_{\rm comm}$ is defined for $\tilde A$ as in \cref{eq:lambda}.
    Using the degree-$n$ Laurent polynomial $Q(z)$ constructed in \cref{thm:arcsin-constru}, we have
    $ \|Q(e^{i \tilde A}) - \frac{2}{\pi} \tilde A \| \leq \tilde\eps,$
    where $n=O(\log(1/\tilde{\eps}))$.
    Then we use GQSP to implement a block encoding of $Q(e^{i \tilde A})$, requiring $n$ queries to controlled-$U$ and controlled-$U^\dagger$. Replacing $U$ with $S^r(-1/r)$, we obtain an approximate block encoding of $2A/(\pi\alpha)$, with the error bounded as
    $O\left(n\eta +\tilde\eps \right).$
    This yields a block encoding of $A$ with normalization factor $\alpha\pi/2$ and error
    $O\left(\alpha n\eta +\alpha\tilde\eps \right).$
    In order to obtain an $(\alpha\pi/2,1,\eps)$-block encoding of $A$, we choose $\tilde\eps, \eta$ such that $ \alpha n\eta =\alpha\tilde\eps=\eps$, namely, $ \eta={\eps}/{\alpha n},\tilde\eps={\eps}/{\alpha}$.
    Therefore, we have
    \begin{equation*}
        n=O\left(\log(\alpha/{\eps})\right),
        \qquad r=O\left( \big(\tilde\lambda_{\rm comm}(\alpha/\eps)\log(\alpha/\eps)\big)^{\frac{1}{2k}} \right)
    \end{equation*}
    and $\tilde\lambda_{\rm comm}=\lambda_{\rm comm}/\alpha^{p+1}$ has a trivial upper bound $\tilde\lambda_{\rm comm} \leq 2^{2k}$.
    Hence, the circuit depth is
    \begin{equation*}
        O\left(n r \Upsilon L\right)=O\left(5^{k-1} L \log\left(\frac{\alpha}{\eps}\right) \left(\frac{\alpha}{\eps}\log\left(\frac{\alpha}{\eps}\right)\right)^{\frac{1}{2k}} \right)
        =\widetilde O\left(5^{k-1} L \left(\alpha/\eps\right)^{\frac{1}{2k}} \right).
        \qedhere
    \end{equation*}
\end{proof}

\subsection{Depth-efficient block encodings via multiproduct formulas}


We next present a second implementation based on multiproduct formulas~\cite{low2019well,aftab2024multi,mizuta2026commutator}. Compared with higher-order product formulas, this achieves a much better dependence on the target precision and avoids the $5^{k-1}$ constant factor. In our setting, this yields a block-encoding construction whose circuit depth depends only polylogarithmically on $1/\varepsilon$, at the cost of using $O(\log\log(1/\varepsilon))$ ancilla qubits.\footnote{In particular, when $\varepsilon > 1/2^L$, this is better than the $O(\log L)$ overhead of the standard LCU framework.}

\subsubsection{Multiproduct formulas}

Let $A=\sum_{\gamma=1}^{L} H_\gamma$.
A $2$nd-based multiproduct formula (MPF) is defined as
\be 
\label{2nd multiproduct formula}
U(\bm{a}, \bm{k}, t) := \sum_{j=1}^m a_j S_2(t/k_j)^{k_j},
\ee
where $S_2(t)$ is the second-order Trotter-Suzuki formula, defined as in \cref{eq:2nd-trotter}.
In~\cite{low2019well,aftab2024multi}, it was proved that there exist $\bm{a}, \bm{k}\in \mb R^m$ with $\|\bm{a}\|_1 = O(\log m), \|\bm{k}\|_1 = O(m^2 \log m)$,
such that
$U(\bm{a}, \bm{k}, t)=e^{-i A t}+O\left(t^{2 m+1}\right),$
achieving $2m$-th order of convergence for any $m\in\mb N$.
They also proved that~\cite[(107), (111)]{aftab2024multi} choosing
\[
m = \left\lceil \frac{1}{2} \log \left(\frac{\mu T}{\varepsilon}\right) \right\rceil,
\quad
r = O\left(\mu T \Big(\frac{\mu T}{\varepsilon}\Big)^{1/(2m)}\right)
= O\left(\mu T \right)
\]
yields
\[
\left\|U\Big(\bm{a}, \bm{k}, \frac{T}{r}\Big)^r - e^{-iAT}\right\| \leq \varepsilon,
\]
where $\mu$ is the commutator scaling defined as below. We summarize their main results below.

\begin{lem}[Theorem 9 and Corollary 10 of \cite{aftab2024multi}]
\label{lem:mpf}
    Consider the Hamiltonian simulation problem up to time $T$ using 2nd-based MPF of convergence order $2 m$. Then there is a quantum circuit that gives a coherent $\eps$-approximate implementation of $e^{-iAT}$, using 
    \[
    O\left(m^2(\log m)^2 \mu_m T\left(\frac{\mu_m T}{\eps}\right)^{1 /(2 m)}\right)
    \]
    queries to the controlled version of $S_2$, and $\lceil\log m\rceil$ ancilla qubits.\footnote{The MPF is implemented via LCU together with robust oblivious amplitude amplification, without intermediate measurement or postselection.} Here
    \begin{align*}
    &\mu_m=\sup _{\substack{j \in 2 \mathbb{Z}^{+}, j \geq 2 m\\ 1 \leq l \leq m}} \Bigg(\sum_{\substack{j_1, \cdots, j_l \in 2 \mathbb{Z}^{+}\\ j_1+\cdots+j_l=j}}\left(\prod_{\kappa=1}^l \lambda_{\mathrm{comm }, j_\kappa+1}\right)\Bigg)^{\frac{1}{j+l}} , \\
    &\lambda_{\mathrm{comm }, j}=\sum_{\gamma_1,\ldots, \gamma_{j}=1}^L \big\| [H_{\gamma_{1}},[H_{\gamma_{2}},[ \cdots ,[H_{\gamma_{j-1}},H_{\gamma_{j}}] \cdots]]] \big\|.
    \end{align*}
    Moreover, suppose that there exists an integer $M$ such that $\sup _{m \geq M} \mu_m \leq \mu$.\footnote{Note that $\mu_m$ has a trivial upper bound $\mu_m \leq 4 \sum_{\gamma=1}^{L}\left\|H_\gamma\right\|$. In Subsection~\ref{subsec:be-with-mpf}, we will only use this trivial estimate for simplicity.} Then there is a quantum circuit that gives a coherent $\eps$-approximate implementation
    of $e^{-iAT}$, using
    \[
    O\left(\mu T \left(\log \left(\frac{\mu T}{\eps}\right)\right)^2 \left(\log\log \left(\frac{\mu T}{\eps}\right)\right)^2 \right)
    \]
    queries to the controlled version of $S_2$, and at most $\lceil\log\log(\mu T/\eps)\rceil$ ancilla qubits.
\end{lem}

Recent work by Mizuta~\cite{mizuta2026commutator} gives a refined commutator-scaling analysis of multiproduct formulas, leading to improved system size dependence for local Hamiltonians. We use simpler MPF bounds in this work, but this refinement might further improve the complexity bounds in the MPF-based setting.


\subsubsection{Block encoding using multiproduct formula}
\label{subsec:be-with-mpf}

\begin{thm}
\label{thm:MPF-general}
    Let $A=\sum_{j=1}^L \alpha_j H_j$ be a Hermitian decomposition with $\alpha_j \in \mathbb{R}_{\geq 0}$ and $\|H_j\|=1$. Let $\alpha=\sum_j \alpha_j$. For any integer $a\geq 1$, there is a quantum circuit that implements an $(\alpha\pi/2, a, \eps)$-block-encoding of $A$. The cost and the circuit depth are
    $ \widetilde{O} (4^a L ({\alpha}/{\eps} )^{\frac{1}{2^{a}}}  )$,
    assuming Hamiltonian simulation of each $H_j$ costs $O(1)$.
\end{thm}
\begin{proof}
    Let $\tilde{A}=A/\alpha$.
    Using the degree-$n$ Laurent polynomial $Q(z)$ constructed in \cref{thm:arcsin-constru}, we have
    $ \|Q(e^{i \tilde A}) - \frac{2}{\pi} \tilde A \| \leq \tilde\eps,$
    where $n=O(\log(1/\tilde{\eps}))$.

    Applying~\cref{lem:mpf} to the normalized $\tilde A=A/\alpha$ with $T=1$ and accuracy $\eta$, we obtain an $\eta$-approximate block encoding of $e^{i\tilde{A}}$ with circuit depth of $O\big(m^2(\log m)^2 L (1/{\eta})^{1 /(2 m)}\big)$ and $O(\log m)$ ancilla qubits.
    Here we used the trivial commutator-scaling bound for $\tilde A=A/\alpha$, which gives a constant commutator-scaling parameter.
    Then we use GQSP to implement a block encoding of $Q(e^{i \tilde A})$, requiring $n$ queries to controlled-$e^{\pm i \tilde A}$.
    Replacing each controlled-$e^{\pm i \tilde A}$ with amplified MPF implementations, we obtain an approximate block encoding of $2A/(\pi\alpha)$, with the error bounded as
    $O\left(n\eta +\tilde\eps \right).$
    This yields a block encoding of $A$ with normalization factor $\alpha\pi/2$ and error
    $O\left(\alpha n\eta +\alpha\tilde\eps \right).$
    In order to obtain an $\eps$-block encoding of $A$, we choose $\tilde\eps, \eta$ such that $ \alpha n\eta =\alpha\tilde\eps=\eps$, namely, $\eta={\eps}/{\alpha n}, \tilde\eps={\eps}/{\alpha}$.
    Therefore, we have
    \[
        n=O\left(\log(\alpha/{\eps})\right),
        \qquad \eta=O\left(\frac{\eps}{\alpha \log(\alpha/\eps)} \right).
    \]
    Let the number of ancilla qubits be $a=\lceil\log m\rceil+1$, then the circuit depth is
    \[
    O\left(n m^2(\log m)^2  L (1/{\eta})^{1 /(2 m)}\right)=O\left(L \log\left(\frac{\alpha}{\eps}\right) 4^a \left(\frac{\alpha \log(\alpha/\eps)}{\eps}\right)^{1/{2^a}}\right)
    =\widetilde{O}\left(4^a L \left(\frac{\alpha}{\eps}\right)^{\frac{1}{2^{a}}} \right).
    \qedhere
    \]
\end{proof}

\begin{cor}
\label{thm:MPF}
    Let $A=\sum_{j=1}^L \alpha_j H_j$ be a Hermitian decomposition with $\alpha_j \in \mathbb{R}_{\geq 0}$ and $\|H_j\|=1$. Let $\alpha=\sum_j \alpha_j$.
    Then there is a quantum circuit that implements an $(\alpha\pi/2, O(\log\log(\alpha/\eps)), \eps)$-block-encoding of $A$. The cost and the circuit depth are $\widetilde{O}(L)$,
    assuming Hamiltonian simulation of each $H_j$ costs $O(1)$.
\end{cor}
\begin{proof}
    Combining~\cref{lem:mpf} and~\cref{thm:MPF-general}, we choose $m = \left\lceil \frac{1}{2} \log \left({1}/{\eta}\right) \right\rceil,$ then the number of ancilla qubits is $a=\lceil\log m\rceil+1=O(\log\log(1/\eta))=O(\log\log(\alpha/\eps))$, and the circuit depth is
    \[
    O\left( L \log\left(\frac{\alpha}{\eps}\right) \left(\log \left(\frac{\alpha \log(\alpha/\eps)}{\eps}\right)\right)^2 \left(\log\log \left(\frac{\alpha \log(\alpha/\eps)}{\eps}\right)\right)^2\right)=\widetilde{O}(L).
    \qedhere
    \]
\end{proof}

In the end, we summarize in~\cref{table:block-encoding-comparison} the resource costs for constructing block encodings.
The comparison focuses on circuit depth and ancilla overhead.
Here LCU denotes the standard linear-combination-of-unitaries construction, while the other two methods are based on Hamiltonian simulation from a Hermitian decomposition. 

\begin{table}[ht!]
\centering
\renewcommand{\arraystretch}{1.4}
\resizebox{\textwidth}{!}{
\begin{tabular}{|c|c|c|c|c|}
\hline
\textbf{Method} & \textbf{Circuit depth} & \textbf{Ancilla} & \textbf{Normalization factor} & \textbf{Error}\\
\hline
LCU
& $O(L)$
& $\log(L)$
& $\alpha$
& $0$ \\
\hline
Higher-order Trotterization (Thm.~\ref{thm:be-with-trotter})
& $\widetilde{O}\big( 5^{k-1} L \left(\alpha/\eps\right)^{\frac{1}{2k}}\big)$
& 1 
& $\alpha\pi/2$
& $\eps$ \\
\hline
Multiproduct formulas (Thm.~\ref{thm:MPF-general})
& $\widetilde{O} (4^a L \left({\alpha}/{\eps}\right)^{\frac{1}{2^{a}}}  )$
& $a$
& $\alpha\pi/2$
& $\eps$ \\
\hline
Multiproduct formulas (Cor.~\ref{thm:MPF})
& $\widetilde{O}(L)$
& $O(\log\log(\alpha/\eps))$
& $\alpha\pi/2$
& $\eps$ \\
\hline
\end{tabular}}
\caption{Comparison of different methods for constructing block encodings of  $A=\sum_{j=1}^L \alpha_j H_j$.}
\label{table:block-encoding-comparison}
\end{table}

\subsection{Low-ancilla QSVT}

One advantage of our constructions is that they yield coherent block encodings. Therefore, they can be directly used as input oracles in the standard QSVT framework. 

We first recall the following standard QSVT result.
\begin{lem}[Theorem 56 of \cite{gilyen2019qsvt}]
\label{lem:qsvt}
    Suppose that $U$ is an $(\alpha, a, \varepsilon)$-block encoding of a Hermitian matrix $A$. Let $P \in \mathbb{R}[x]$ be a degree-$d$ polynomial satisfying $\left|P(x)\right| \leq {1}/{2}$ for all $x \in[-1,1]$.
    Then there is a quantum circuit $\widetilde{U}$, which is a $(1, a+2,4 d \sqrt{\varepsilon / \alpha})$-block encoding of $P(A / \alpha)$, and consists of $d$ applications of $U$ and $U^{\dagger}$, a single application of controlled-$U$ and ${O}((a+1) d)$ other one- and two-qubit gates.
\end{lem}

Combining this with our results in Subsection~\ref{subsec:be-with-trotter}, we obtain a low-ancilla implementation of QSVT for operators expressed as Hermitian decompositions.

\begin{thm}[Low-ancilla QSVT via Trotterization]
\label{thm:low-ancilla-qsvt}
    Assume that $A=\sum_{j=1}^L \alpha_j H_j$ is a Hermitian decomposition with $\alpha_j \in \mathbb{R}_{\geq 0}$ and $\|H_j\|=1$. Let $\alpha=\sum_j \alpha_j$ and $k \in \mathbb{N}$. 
    Let $P(x) \in \mathbb{R}[x]$ be a degree-$d$ polynomial satisfying $\left|P(x)\right| \leq {1}/{2}$ for all $x \in[-1,1]$.
    Then there is a quantum circuit that implements a $(1, 3,\eps)$-block encoding of $P(2A / \alpha\pi)$, with circuit depth $\widetilde{O}(5^{k-1}L d^{1+1/k}/\eps^{1/k})$,
    assuming Hamiltonian simulation of each $H_j$ costs $O(1)$.
\end{thm}

\begin{proof}
    This follows from \cref{thm:be-with-trotter} and \cref{lem:qsvt}. 
    By \cref{thm:be-with-trotter}, we can construct an $(\alpha\pi/2, 1, \eta)$-block-encoding $U$ of $A$, with circuit depth $\widetilde{O}(5^{k-1}L (\alpha/\eta)^{1/2k})$.
    Then applying Lemma~\ref{lem:qsvt} yields a $(1, 3, 4d\sqrt{2\eta/ \alpha\pi})$-block encoding of $P(2A / \alpha\pi)$, using $d$ applications of $U$ and $U^\dagger$. Choose $\eta=\Theta(\alpha\eps^2/d^2)$ so that $4d\sqrt{2\eta/ \alpha\pi}\leq \eps$.
    The resulting circuit depth is therefore $\widetilde O(d 5^{k-1}L(\alpha/\eta)^{1/(2k)})$ $=\widetilde{O}(5^{k-1}L d^{1+1/k}/\eps^{1/k}).$
\end{proof}

\begin{rmk}
    The above QSVT implementation is coherent and yields a block encoding of $P(2A/\alpha\pi)$ with depth $\widetilde O(Ld(d/\varepsilon)^{o(1)})$. If only an expectation value is required, our Hamiltonian-evolution-based construction can instead be combined with the Richardson extrapolation method of~\cite{chakraborty2025quantum}, since it fits their interleaved-sequence framework. This improves the precision dependence to $\polylog(1/\varepsilon)$, giving an incoherent estimator with depth $\widetilde O(Ld^{1+o(1)})$.
\end{rmk}

\begin{rmk}
\label{rmk:qsvt-mpf}
    The same argument applies to the multiproduct-formula construction of Subsection~\ref{subsec:be-with-mpf}. Under the assumptions of~\cref{thm:low-ancilla-qsvt}, choosing the block-encoding accuracy $\eta=\Theta(\alpha\varepsilon^2/d^2)$ in~\cref{thm:MPF} yields an $\varepsilon$-approximate block encoding of $P(2A/\alpha\pi)$ with depth $\widetilde O(Ld)$ and ancilla overhead $O(\log\log(\alpha d/\varepsilon))$. Thus the precision dependence becomes polylogarithmic at the cost of additional ancillas.

\end{rmk}

Taken together, these results show that our coherent block encodings can be directly combined with standard QSVT to implement polynomial transformations of $A$. This enables a broad range of spectral transformations beyond real-time Hamiltonian simulation, including those used for ground-state preparation, Green's function estimation, and thermal-state preparation~\cite{gilyen2019optimizing,martyn2021grand}.

\section{Proof of Theorem \ref{thm:tradeoff}}
\label{sec:proof of thm 1}

The proof uses a deliberately simple family of Hermitian involutions. The point is not to make the Hamiltonian itself difficult to simulate by every conceivable algorithm, but rather to expose a universal linear-response obstruction within the specified simulator architecture.

Let $\delta$ be a parameter; define
\[
H_1(\delta) = -(\cos\delta) Z + (\sin\delta) X,
\quad
H_j=Z ~ (j\geq 2).
\]
For coefficients $\alpha_1=1/2$ and $\sum_{j=2}^L \alpha_j=1/2$, the target Hamiltonian is
\be \label{eq:target Hamiltonian}
A(\delta)=\sum_{j=1}^L \alpha_j H_j = \frac{1}{2} H_1(\delta) + \frac{1}{2}Z.
\ee
So $A(0)=0, A'(0)=X/2$. For the target time evolution $U(t,\delta):=e^{-it A(\delta)}$, we have $U(t,0)=I$ and $\partial_\delta U(t,0)=-it X/2$.

Our main idea is as follows.
Consider $C(t,\delta)$ defined in \eqref{intro eq:C(t)}, now viewed as a function of $\delta$. By construction, $C(t,0)=\eta I$, so the zeroth-order condition at $\delta=0$ is satisfied exactly. The first nontrivial necessary condition is therefore to reproduce the correct linear response,
$$
\partial_\delta C(t,0)\approx
\eta\,\partial_\delta U(t,0)
=
-\frac{it\eta}{2}X.
$$
As we show below, the linear response of the simulator can be written as
\be \label{eq:def of g(t)}
\partial_\delta C(t,0)
=
-\frac{it\eta}{2}X+g(t)X,
\ee
for some function $g(t)$.
Thus, $g(t)$ measures the mismatch between the linear response of the simulator and that of the target evolution. The linear response therefore serves as a witness to the limitations of the architecture: by showing that $g(t)$ cannot be made arbitrarily small without sufficient ancilla space or circuit depth, we obtain the desired ancilla-depth tradeoff. Our analysis below is devoted to quantifying this unavoidable linear-response error. The analysis relies on two key estimates for $g(t)$, established in \eqref{eq: derivative} and \eqref{eq: derivative upper bound} below.

\subsection{A lower bound of $|g^{(2n+1)}(0)|$}

For $C(t,\delta)$ defined in \eqref{intro eq:C(t)}, we first derive an explicit formula for $g(t)$.
For the particular Hamiltonian \eqref{eq:target Hamiltonian}, the second-order Suzuki formula \eqref{eq:2nd-trotter} gives
\[
S_2(h,\delta)
:=\prod_{\gamma=L}^1 e^{-i h \alpha_\gamma H_\gamma  / 2} \prod_{\gamma=1}^{L} e^{-i h \alpha_\gamma H_\gamma  / 2}
=e^{-ihZ/4} e^{-ih H_1(\delta)/2} e^{-ihZ/4}.
\]
Moreover, a simple calculation shows that $S_2(h,0)=I$ and $\partial_\delta S_2(h,0)=-i \sin(h/2)X$ for every $h$. We define
\[
\rho_1(h) = \frac{\partial_\delta S_2(h,0)}{-i h X} =  \frac{\sin(h/2)}{h} = 
\sum_{n=0}^\infty \frac{(-1)^n}{2^{2n+1} (2n+1)!} h^{2n}.
\]
From the recursive Suzuki formula \eqref{eq:trotter-suzuki}, if we set
$\rho_k(h) = \frac{\partial_\delta S_{2k}(h,0)}{-i h X}$, then we obtain the following recursive formula, whose proof is straightforward.

\begin{lem}
\label{lem:recursive formula}
For all $k\geq 2$ and all $h$, we have $S_{2k}(h,0)=I$  and
\[
\rho_k(h) = 4u_k \rho_{k-1}(u_k h)+(1-4u_k) 
\rho_{k-1}((1-4u_k) h),
\]
where $u_k=(4-4^{1/(2k-1)})^{-1}$.
\end{lem}

\begin{proof}
We proceed by induction on $k$. Recall that the recursive Suzuki formula is defined by
\[
S_{2k}(h,\delta)
=
S_{2k-2}(u_k h,\delta)^2\,
S_{2k-2}((1-4u_k)h,\delta)\,
S_{2k-2}(u_k h,\delta)^2.
\]
Thus $S_{2k}(h,0)=I$ follows immediately by induction because $S_2(h,0)=I$.
Next, differentiate the recursive formula with respect to $\delta$ and by induction on $k$, we obtain
\beas
\partial_\delta S_{2k}(h,0)
&=&
4 \partial_\delta S_{2k-2}(u_k h,0)
+ \partial_\delta S_{2k-2}((1-4u_k) h,0) \\
&=& -4iu_kh 
\rho_{k-1}(u_kh)X 
-i(1-4u_k)h 
\rho_{k-1}((1-4u_k)h)X \\
&=& 
-ih
\Big[
4u_k\rho_{k-1}(u_kh)
+
(1-4u_k)
\rho_{k-1}((1-4u_k)h)
\Big]X.
\eeas
Comparing this with the definition $\partial_\delta S_{2k}(h,0)
=
-ih \rho_k(h)X$
proves the claim.
\end{proof}
We now write
\be \label{eq:rhok}
\rho_k(h) = \frac{1}{2} + \sum_{n=k}^\infty c_{k,n} h^{2n},
\ee
then by Lemma \ref{lem:recursive formula},
\be \label{eq:recursive formula of ckn}
c_{k,n} = c_{k-1,n} \Big[4 u_k^{2n+1} + (1-4u_k)^{2n+1} \Big] .
\ee
Since $c_{1,n}\neq0$ and $4 u_k^{2n+1} + (1-4u_k)^{2n+1} \neq 0$ for $n\geq k$, we have that $c_{k,n}\neq 0$ for all $n\ge k$.

Note that for a branch using $m$ refinement steps, $P_m(t,\delta):=S_{2k}(t/sm,\delta)^m$, we have
\[
P_m(t,0)=I,\quad 
\partial_\delta P_m(t,0)
=-\frac{i t}{s} \rho_k(t/sm) X.
\]
With the above preliminaries, we now compute $g(t)$. For any $\ell, r$, define
\[
\ket{a_\ell} = V_{\ell-1}\cdots V_0 \ket{0^a},
\quad 
\bra{b_\ell} = \bra{0^a} V_s \cdots V_{\ell},
\quad
\omega_{\ell r} = \braket{b_\ell}{r} \braket{r}{a_\ell}.
\]
Then $\sum_{r=1}^R \omega_{\ell r} = \bra{0^a} V_s\cdots V_0 \ket{0^a}=\eta$, where $R=2^a$. Moreover,
\[
\partial_\delta C(t,0) = -\frac{i t}{s} X
\sum_{\ell=1}^s \sum_{r=1}^R \omega_{\ell r} \rho_k(t/sm_{\ell r}) .
\]
Substituting the expression \eqref{eq:rhok} of $\rho_k$ into the above yields
\beas
\partial_\delta C(t,0) 
&=& -\frac{i t}{s} X
\sum_{\ell=1}^s \sum_{r=1}^R \omega_{\ell r} 
\left(
\frac{1}{2} + \sum_{n=k}^\infty c_{k,n} (t/s)^{2n} m_{\ell r}^{-2n}
\right) \\
&=& -\frac{i \eta t}{2} X
-it X \sum_{n=k}^\infty c_{k,n} (t/s)^{2n}  M_n,
\eeas
where
\[
M_n = \frac{1}{s} \sum_{\ell=1}^s \sum_{r=1}^R \omega_{\ell r} m_{\ell r}^{-2n}.
\]
We denote by $g(t)$ the scalar coefficient multiplying $X$ in the second term of $\partial_\delta C(t,0)$. 

We next prove that there exists an $n$ for which $|g^{(2n+1)}(0)|$ is not too small. Since $g^{(2n+1)}(0) = -i (2n+1)! c_{k,n} s^{-2n} M_n$, it remains to obtain lower bounds on $|c_{k,n}|$ and $|M_n|$. For $c_{k,n}$, the recursive relation in \eqref{eq:recursive formula of ckn} gives the following bound.

\begin{lem}
\label{lem:bound of ckn}
For any $n\geq k \geq 1$, we have
$(2n+1)! |c_{k,n}| \geq e^{-O(kn+k\log k)}.$
\end{lem}

\begin{proof}
We will use the recursive formula \eqref{eq:recursive formula of ckn}. When $k=1$, we have that $(2n+1)!|c_{1,n}|=1/2^{2n+1}=e^{-O(n)}$. By \eqref{eq:recursive formula of ckn}, we have $|c_{r,n}| = |c_{r-1,n}| A_{r,n}$ with $A_{r,n}=u_r^{2n+1} (4^{(2n+1)/(2r-1)} - 4)$ via some simplification. Since $u_r\geq 1/3$ for $r\geq 2$ and $n\geq k\geq r$, we have
$4^{(2n+1)/(2r-1)} - 4 \geq 4 (4^{2/(2r-1)}-1) \geq c/r$ for some constant $c$. Hence, $|A_{r,n}|\geq 3^{-(2n+1)} c/r$. Therefore,
\[
(2n+1)! |c_{k,n}| = (2n+1)! |c_{1,n}|  |A_{2,n} \cdots A_{k,n}|
\geq 2^{-(2n+1)} (3^{-(2n+1)} c )^{k-1}/k!.
\]
Using $k! \leq e^{O(k \log k)}$ and taking the logarithm leads to the claimed result.
\end{proof}

We next derive a lower bound on $|M_n|$.




\begin{lem}
\label{bound of Mn}
There exists $n\in\{k,k+1,\ldots,k+sR-1\}$ such that
\[
|M_n| \geq \frac{|\eta|}{sR Q^{2k} (2Q^4)^{sR-1}}.
\]
In particular,
\[
|M_n| \geq |\eta| e^{-O((k+sR) \log (sRQ))}.
\]
\end{lem}

\begin{proof}
Combining all terms with the same value of $m_{\ell r}$,
we may write $M_n=\sum_{j=1}^K d_j z_j^{n}$,
where $z_1,\ldots,z_K$ are distinct, and $K\leq sR$.
Let \(m_1,\ldots,m_K\) denote the distinct values among the \(m_{\ell r}\)'s, then $z_j=m_j^{-2}$ for $j=1, \ldots, K$. Moreover, $\sum_j d_j=\eta$.
Let $P(z)$ be the unique polynomial of degree at most $K-1$ satisfying $P(z_j)=z_j^{-k}$. In Lagrange form
\[
P(z) = \sum_{j=1}^K z_j^{-k}
\prod_{q\neq j} \frac{z-z_q}{z_j-z_q}.
\]
Write $P(z)=\sum_{j=0}^{K-1} a_j z^j$. Then
\[
\eta = \sum_{j=1}^K d_j = 
\sum_{j=1}^K d_j z_j^k P(z_j)
= \sum_{q=0}^{K-1} a_q \left(\sum_{j=1}^K d_j z_j^{q+k}\right) = \sum_{q=0}^{K-1} a_q M_{k+q}.
\]
This implies $|\eta| \leq \sum_{q=0}^{K-1} |a_q| |M_{k+q}|\leq \sum_{q=0}^{K-1} |a_q|  \max_q |M_{k+q}|$. Thus, it suffices to bound the $\ell_1$ norm of the coefficients $\|P\|_1=\sum_{k=0}^{K-1} |a_k|$.
It is easy to see that $\|\prod_{q\neq j} (z-z_q)\|_1 \leq \prod_{q\neq j} \| z-z_q\|_1
\leq \prod_{q\neq j} (1+|z_q|)\leq 2^{K-1}$.
From \eqref{intro: eq of Q}, we know that $m_{\ell r}\leq Q$, i.e., $z_j\geq Q^{-2}$. This shows that 
$|z_j^{-k}|\leq Q^{2k}$ and
$\prod_{q\neq j}  |z_j-z_q| \geq Q^{-4(K-1)}$.
Putting it all together, $\|P\|_1\leq K Q^{2k} (2Q^4)^{K-1}$.
\end{proof}

Combining Lemmas~\ref{lem:bound of ckn} and \ref{bound of Mn} with the bounds $s\le Q$, $n\le k+sR$, and $k=O(\log Q)$, we obtain
\be \label{eq: derivative}
|g^{(2n+1)}(0)|
\geq |\eta| e^{-O((k+sR) \log (sRQ))}.
\ee

\subsection{An upper bound of $|g^{(2n+1)}(0)|$}

We now derive an upper bound on $|g^{(2n+1)}(0)|$ from a different perspective. Let
$
E(t,\delta)=C(t,\delta)-\eta U(t,\delta)
$
denote the error operator. By \eqref{eq:def of g(t)} and the identity $\partial_\delta U(t,0)=-itX/2$, we have
\be \label{eq:a new formula of g(t)}
g(t)X
=
\left.\partial_\delta E(t,\delta)\right|_{\delta=0}
=
\left.\partial_\delta\bigl(C(t,\delta)-\eta U(t,\delta)\bigr)\right|_{\delta=0}.
\ee
This relation, together with a suitable error analysis, yields an upper bound on $|g(t)|$. We then apply the following result recursively to bound $|g^{(2n+1)}(0)|$.

\begin{lem}[Theorem 9.1 of \cite{Erdelyi2017}]
\label{thm:Markov-exponential-sum}
Let $f(t)=bt+\sum_{j=1}^{M} a_j e^{i\lambda_j t},$
where $\lambda_1,\ldots,\lambda_M$ are distinct real numbers. Then, there is a constant $C>1$ 
such that \footnote{
In the original Theorem 9.1, one has $b=0$. The present result follows easily from the approximation
$
t=\lim_{\delta\to 0}\frac{e^{i\delta t}-1}{i\delta}.
$
Indeed, consider
$
f_\delta(t)
=
b\,\frac{e^{i\delta t}-1}{i\delta}
+
\sum_{j=1}^{M} a_j e^{i\lambda_j t}.
$
Applying Theorem 9.1 of \cite{Erdelyi2017} to $f_\delta$ and then letting $\delta\to 0$ yields the claimed bound.
}
\be
\|f'\|_{L^\infty([0,1])}
\leq
C
\left(
108(M+2)^5+\sum_{j=1}^{M}\lambda_j^2
\right)^{1/2}
\|f\|_{L^\infty([0,1])}.
\label{Markov exponential sum}
\ee
\end{lem}

Recall that for the hard family,
\[
A(\delta)
=
\frac12(\sin\delta)\,X
+
\frac12(1-\cos\delta)Z,
\]
which is a matrix-valued trigonometric polynomial of degree one.
Moreover, for real $\delta$, $\|A(\delta)\|\leq 1$ and $\|\partial_\delta A(\delta)\|\leq 1/2$.
For an integer $J \geq 0$, define the truncated exponential of the ideal evolution $U(t,\delta)=e^{-itA(\delta)}$ by
\[
U_J (t,\delta)
=
\sum_{q=0}^{J }
\frac{(-it)^q}{q!}A(\delta)^q.
\]
The trigonometric degree satisfies $\deg_\delta U_J(t,\cdot)\leq J.$
For $|t|\leq1$, the Taylor remainder satisfies
\[
\|U(t,\delta)-U_J(t,\delta)\|
\leq
\sum_{q=J+1}^{\infty}\frac{1}{q!}
\leq
\frac{C}{(J+1)!}
\]
for some constant $C$.
Furthermore,
\[
\partial_\delta A(\delta)^q
=
\sum_{j=0}^{q-1}
A(\delta)^j
\partial_\delta A(\delta)
A(\delta)^{q-1-j},
\]
and therefore
$\|\partial_\delta A(\delta)^q\|
\leq
q\|\partial_\delta A(\delta)\|
\leq
q/2$.
It follows that
\[
\|\partial_\delta(U-U_J)\|
\leq
\frac12
\sum_{q=J+1}^{\infty}
\frac{q}{q!}
\leq
\frac{C}{J!}.
\]
Consequently, choosing $J=O(\log 1/\varepsilon)$
gives
$\|U-U_J\|_\infty
+
\|\partial_\delta(U-U_J)\|_\infty
\leq
O(\varepsilon),$
uniformly for real $\delta$ and $|t|\leq1$.

Suppose the architecture approximates the ideal evolution uniformly; then
the error
$E(t,\delta) = C(t,\delta)-\eta U(t,\delta)$
satisfies $\|E(t,\delta)\|\leq \varepsilon$ for all real $\delta$ and $|t|\le 1$.

\begin{lem} 
\label{lem:bound of gt}
Under the assumptions above, we have
\[
\sup_{|t|\leq 1} |g(t)| = O\big(\varepsilon (Q+1+\log (1/\varepsilon) )\big).
\]
\end{lem}

\begin{proof}
From \eqref{eq:a new formula of g(t)}, it suffices to bound $\| \partial_\delta E(t,\delta)|_{\delta=0}\|$.
Consider $F(t,\delta):=C(t,\delta)-\eta U_J(t,\delta)$, then 
\[
\|F(t,\delta)\|  \leq 
\| C(t,\delta)-\eta U(t,\delta)  \|  
+ |\eta| \| U(t,\delta) - U_J(t,\delta) \| = O(\varepsilon).
\]
By the trigonometric Bernstein inequality \cite[Chapter VI, Section 23]{Bary1964},
$
\sup_\delta \|\partial_\delta F(t,\delta)\|  \leq 
\deg_\delta(F) \sup_\delta  \|F(t,\delta)\|  = 
O((Q+J) \varepsilon).
$\footnote{As matrix-valued trigonometric polynomials in $\delta$, $H_1(\delta)=-(\cos\delta) Z + (\sin\delta) X$ and $e^{-i \tau H_1(\delta)}=
(\cos \tau) I - i (\sin \tau) H_1(\delta)$ both have degree at most 1. For $[S_{2k}(t/sm_{\ell r})]^{m_{\ell r}}$, $H_1(\delta)$ appears at most $O(5^{k-1} m_{\ell r})$ times. Thus, there appears at most $O(5^{k-1} \sum_\ell \max_r m_{\ell r})=O(Q)$ times for $C(t,\delta)$. This implies that the trigonometric degree of $C(t,\delta)$ as a function of $\delta$ is $O(Q)$. Thus $\deg_\delta(F)=O(Q+J)$.}
Thus
\[
\|\partial_\delta E(t,\delta)|_{\delta=0} \| 
\leq \|\partial_\delta F(t,\delta)|_{\delta=0} \| 
+ |\eta| \|\partial_\delta (U-U_J)|_{\delta=0} \| 
=O((Q+J+1) \varepsilon).
\qedhere
\]
\end{proof}

To apply Lemma \ref{thm:Markov-exponential-sum}, we first bound the number and magnitudes of the frequencies appearing in $g(t)$.
Note that $h\rho_1(h)=\sin(h/2)$ contains one sine frequency, and from Lemma \ref{lem:recursive formula}, it is easy to see that $h \rho_k(h)$ contains at most $2^{k-1}$ sine frequencies. Recall that $g(t)$ roughly looks like $g(t)=-it (\sum_{j=1}^K d_j \rho_k(t/sm_j)-\eta/2)$, there are at most $2^{k-1}K$ sine frequencies. 
A frequency $|\lambda|$ in $g(t)$ is rescaled to $\lambda/sm_j$ which is bounded by $|\lambda|$. For $t \rho_1(t)$, the frequency $|\lambda|\leq 1/2$. Thus, all nonzero frequencies appearing in the exponential-sum representation
of $g(t)$ have absolute value at most $1/2$. 
Thus $M=O(2^k K)=O(2^ks R)$ and $|\lambda_j| \leq 1/2$. Applying Lemma~\ref{thm:Markov-exponential-sum}, we obtain
\[
\|g'\|_{L^\infty([0,1])}
\leq
C
\left(
108(M+2)^5+M/4
\right)^{1/2}
\|g\|_{L^\infty([0,1])}
=O((2^ks R)^{5/2}) \|g\|_{L^\infty([0,1])}.
\]
Applying this $r$ times, we obtain
\be \label{eq: derivative upper bound}
\|g^{(r)}\|_{L^\infty([0,1])}
\leq O((2^ks R)^{5r/2}) \|g\|_{L^\infty([0,1])}
=e^{O(rk+r \log (sR))} \|g\|_{L^\infty([0,1])} .
\ee
Combining Equations \eqref{eq: derivative}, \eqref{eq: derivative upper bound} with $r=2n+1$, and
Lemma \ref{lem:bound of gt},
we obtain the claimed result in Theorem \ref{thm:tradeoff} via some simplification.

\section{Proof of Theorem \ref{thm:balancedcase} }

\label{sec:proof of thm2}

Theorem \ref{thm:balancedcase} is a special case of the following more general case of arbitrary coefficients. We assume $\lambda_j>0$ and $\sum_j \lambda_j=1$.
Let $\tau > \varepsilon$, define the effective support above the approximation noise floor as
\[
J_\tau = \{j: \lambda_j \geq \tau\}, \quad
M_\tau = |J_\tau|.
\]
The quantity $M_\tau$ may be interpreted as the effective number
of logical terms that remain visible above the approximation noise
floor.  Coefficients below $\varepsilon$ cannot in general
be resolved individually by an $\varepsilon$-approximate
implementation, which explains the appearance of the threshold
$\tau>\varepsilon$.

\begin{thm}[Effective-support ancilla-query tradeoff]
\label{thm:General effective-support tradeoff}
Assume \eqref{intro:circuit goal} holds for all unitaries $P_1,\ldots,P_L\in U(N)$ with $\lambda_j\geq 0$ and $\sum_j \lambda_j=1$. Assume $N\geq T(M_\tau+1)^2$. Then
\be
2^a \geq \frac{M_\tau^2}{M_\tau + 2K^2_{\varepsilon/\tau}(T-1)}.
\ee
In particular, if $\varepsilon\leq r \tau$ for some fixed constant $r<1$, then
\be
2^a \geq \Omega_r\left( \frac{M_\tau^2}{M_\tau + T} \right).
\ee
When $\varepsilon=0$, we have $2^a \geq  |\supp(\lambda)|$ for arbitrary $T$.\footnote{
Here $K_{\varepsilon/\tau}$ is defined via \eqref{intro:eq:K-constant} and $\supp(\lambda)=\{j:\lambda_j>0\}$.}
\end{thm}

Taking $\lambda_j=1/L$ and $\tau=1/L$ in Theorem \ref{thm:General effective-support tradeoff} gives Theorem \ref{thm:balancedcase}. 

\begin{rmk}[Relation to the upper bounds]
The lower bound above concerns a black-box repeated-query model
that is required to work uniformly for arbitrary unitary inputs
$P_1,\ldots,P_L$.  This is different from the Hermitian-decomposition
setting of Section~\ref{sec:Low-ancilla block encodings for Hermitian decompositions}, where the construction has access to
Hamiltonian evolutions generated by Hermitian terms.  
Therefore, Theorems  \ref{thm:balancedcase} and \ref{thm:General effective-support tradeoff} should be viewed as a structural
limitation of the repeated-query LCU architecture, rather than as
a matching lower bound for the Hamiltonian-simulation constructions
of Section~\ref{sec:Low-ancilla block encodings for Hermitian decompositions}.
\end{rmk}

The main idea of the proof is as follows. We first rewrite the repeated-query circuit in \eqref{intro:circuit for LCU}  as a noncommutative polynomial in the input unitaries,
$B(P)=\sum_{r=0}^{T} B_r(P)$,
where $B_r$ is the homogeneous component of degree $r$; see \eqref{eq: decomposition of BP}. Since the target LCU operator
$
A(P)=\sum_{j=1}^L \lambda_j P_j
$
is linear in the input unitaries, uniform $\varepsilon$-approximation implies that 
$B_1(P)\approx A(P)$
and
$B_r(P)\approx 0$ for $r\geq 2.$
We next examine what the latter condition
implies for the individual query paths of the circuit. For each degree $r$, we group together all length-$r$ query paths according to their first and last logical labels, and record their total amplitudes in a matrix $C^{(r)}$; see \eqref{eq:matrix Cr}.
By evaluating $B_r$ on a suitable family of test unitaries, we obtain
$
C^{(r)}_{S,S} \approx 0.
$
This is one of the most technical parts of the proof.
To describe these query paths, we introduce a strict lower-triangular matrix $H$, whose entry $H_{ij}$ records the transition amplitude from query position $j$ to query position $i$. Consequently, the entries of $H^{r-1}$ collect the amplitudes of all length-$r$ query paths with prescribed initial and final query positions. The path expansion of the circuit then gives
$
C^{(r)}_{S,S}
=
D_{\mathrm{out}}\,
U^\dagger H^{r-1}V\,
D_{\mathrm{in}},
$
where the diagonal matrices $D_{\mathrm{out}}$ and $D_{\mathrm{in}}$ are controlled by the degree-one part of the approximation. 
Using the degree-one constraints, a straightforward estimate shows that $U$ and $V$ are close and that
$
U^\dagger H^kU\approx 0
$
for all $1\leq k<T$.

Finally, a Gram-matrix representation associated with the query-history vectors gives
$
I-H^\dagger H\succeq 0
$ and 
$\rank(I-H^\dagger H)\le 2^a.
$
Hence, $H$ can be represented as the compression of a shift operator acting on a larger space, where the required extra dimension is at most $2^a$.
The key ingredient at this stage is Lemma \ref{lem:defect-shift-packing}, which relates the overlaps of shifted copies to the available dimension.
Under this dilation, the matrices $U^\dagger H^k U$ describe the overlaps between different shifted copies of the query-history vectors.
Since these overlaps are small, sufficiently separated shifted copies are nearly orthogonal. Their Gram matrix therefore has large rank. In addition, all these vectors lie in a space of limited dimension, controlled by the ancilla size and the number of query steps. Comparing this bound with the rank lower bound gives the desired ancilla-query tradeoff.

\subsection{From approximation to query-history constraints}

This part contains some preliminary lemmas.
Rewriting $B$ yields the following decomposition.

\begin{lem}
\label{lem:reformulation}
For $B$ defined in \eqref{intro:circuit goal},
there exist unit vectors $\ket{u}, \ket{v}\in (\mathbb{C}^2)^{\otimes a}$ and rank-one projectors
\[
\Pi_j = \ket{\psi_j} \bra{\psi_j},\quad 
A_j = I - \Pi_j, \quad j=1,\ldots,T
\]
such that
\be \label{eq for B}
B = (\bra{u} \otimes I) \prod_{j=T}^1(A_j\otimes I+\Pi_j\otimes P_{\alpha_j}) (\ket{v} \otimes I).
\ee
\end{lem}

\begin{proof}
Let
$
Q=\ket{0^a}\bra{0^a},
C_Q(P):=(I-Q)\otimes I+Q\otimes P,
$
and define
$
W_j:=V_{j-1}\cdots V_0, (j=1,\ldots,T+1).
$
For any ancilla unitary $W$, we have
$
C_Q(P)(W\otimes I)
=
(W\otimes I)\,
C_{W^\dagger QW}(P).
$
Applying this identity successively to the query gates in \eqref{intro:circuit for LCU} gives
$$
U
=
(W_{T+1}\otimes I)
\prod_{j=T}^{1}
C_{\Pi_j}(P_{\alpha_j}),
\qquad
\Pi_j:=W_j^\dagger QW_j .
$$
Since each $\Pi_j$ is a rank-one projector, writing
$
\Pi_j=\ket{\psi_j}\bra{\psi_j},
A_j:=I-\Pi_j,
$
we have
$
C_{\Pi_j}(P_{\alpha_j})
=
A_j\otimes I+\Pi_j\otimes P_{\alpha_j}.
$
Taking the $\ket{0^a}$ block of $U$ and setting
$
\ket{v}:=\ket{0^a},
\ket{u}:=W_{T+1}^\dagger\ket{0^a},
$
therefore yields \eqref{eq for B}.
\end{proof}

Next, we expand \eqref{eq for B}. For convenience, we introduce some notation.
Using the notation in the proof of Lemma \ref{lem:reformulation}, we define
\be
\label{eq:x and y}
x_j := \bra{u} A_T\cdots A_{j+1} \ket{\psi_j},
\quad 
y_j := \bra{\psi_j} A_{j-1}\cdots A_1 \ket{v},
\ee
and for $i\geq j$, 
\[
h_{ij} := \bra{\psi_i} A_{i-1}\cdots A_{j+1} \ket{\psi_j}.
\]
Expanding \eqref{eq for B} gives
\be \label{eq: B in fourier expansion}
B = \sum_{S \subseteq [T]} \alpha_S P_S,
\ee
where $\alpha_{\emptyset}=\bra{u}A_T\cdots A_1 \ket{v}$ and if $S=\{t_1<\cdots<t_m\}$ then
\[
\alpha_S = x_{t_m} h_{t_m t_{m-1}}\cdots h_{t_2t_1} y_{t_1},
\quad 
P_S = P_{\alpha_{t_m}} \cdots P_{\alpha_{t_1}}.
\]

From now on, we make the dependence on $P_1,\ldots,P_L$ explicit by writing $A(P) = \sum_{j=1}^L \lambda_j P_j$, and $B(P)=B$.
For $B(P)$ in \eqref{eq: B in fourier expansion}, we decompose it into its homogeneous components:
\be
\label{eq: decomposition of BP}
B(P)  = \sum_{r=0}^T B_r(P) , \quad 
B_r(P)   =
\sum_{w=(j_1,\ldots,j_r)\in [L]^r} b_w P_w,
\quad P_w=P_{j_r}\cdots P_{j_1}.
\ee
Thus, $B_r(P)$ is the homogeneous component of degree $r$ in the noncommuting unitaries $P_1,\ldots,P_L$.

\begin{lem}
\label{lem:affine-homogeneous-extraction}
Assume $\| A(P)  - B(P)  \| \leq \varepsilon$ for any unitaries $P_1,\ldots,P_L$, 
then $\|B_r(P) \|\leq \varepsilon$ for any $r\geq 2$ and
$\|B_1(P)-\sum_j \lambda_j P_j\|\leq \varepsilon$.
\end{lem}

\begin{proof}
We introduce a complex parameter $z$ of norm 1.
Then the error $E(zP)=B(zP)-z A(P)$ has a Fourier expansion $\sum_{r=0}^T z^r E_r(P)$ with
$E_1(P)=B_1(P)-\sum_j \lambda_j P_j$, and
$E_r(P) = B_r(P) = \frac{1}{2\pi} \int_0^{2\pi} e^{-i r\theta} E(e^{i\theta} P) d\theta$. So $\|E_r(P)\|\leq \varepsilon$ for $r\geq 2$.
\end{proof}

For $r\geq 2$, define matrix $C^{(r)}$ via
\be
\label{eq:matrix Cr}
C_{ji}^{(r)} = \sum_{\substack{w=(j_1,\ldots,j_r)\in[L]^r \\ j_1=i, j_r=j}} b_w.
\ee

\begin{lem} 
\label{lem:principal submatrix}
For a subset $S\subseteq [L]$, let $C_{S,S}^{(r)}$ be the corresponding principal submatrix of $C^{(r)}$.
If $N\geq r (|S|+1)^2$ for $2\leq r \leq T$, then $\|C_{S,S}^{(r)}\|_F \leq \varepsilon$. \footnote{Recall that $N$ is the dimension of the system space, i.e., $P_j$ are unitaries of dimension $N$.}
\end{lem}

\begin{proof}
Denote $M=|S|$ and relabel the elements of $S$ as $\{1,\ldots,M\}$.  We construct a family
of test unitaries whose purpose is to record only the first and
last logical labels of a length-$r$ word.

Consider the auxiliary Hilbert space
$\mathcal K_r
:=
\mathbb C^r_{\mathrm{clock}}
\otimes
\mathbb C^{M+1}_{F}
\otimes
\mathbb C^{M+1}_{R}.$
The first register is a cyclic clock, while the second and third
registers will record the first and last labels of a word,
respectively.  Its dimension is $\dim \mathcal K_r=r(M+1)^2.$
Let $S_r|c\rangle
=
|c+1 \!\!\!\pmod r\rangle$ for $c=0,\ldots,r-1$,
denote the cyclic shift on the clock register.  For each
$j\in S$, let $X_j$ denote the permutation unitary on
$\mathbb C^{M+1}$ which swaps the basis states $|0\rangle$ and
$|j\rangle$ and fixes all other basis states.
For $j\in S$, define the clock-controlled unitary
\[
D_j
:=
|0\rangle\langle 0|\otimes X_j\otimes I
+
|r-1\rangle\langle r-1|\otimes I\otimes X_j
+
\sum_{c=1}^{r-2}
|c\rangle\langle c|\otimes I\otimes I,
\]
where the sum is absent when $r=2$.  Thus, at clock value $0$ the
label $j$ is written into the $F$ register, at clock value $r-1$ it
is written into the $R$ register, and at all intermediate clock
values the endpoint registers are left unchanged.  Set
$W_j^{(r)}
:=
(S_r\otimes I\otimes I)D_j$ for $j\in S$,
and 
define simply
$
W_j^{(r)}
:=
S_r\otimes I\otimes I
$
for $j\notin S$. Each $W_j^{(r)}$ is unitary.

We now evaluate a length-$r$ logical word
$w=(j_1,\ldots,j_r).$
Since operators act from right to left, the first operator applied to
the state $|0,0,0\rangle$ is $W_{j_1}^{(r)}$.  If $j_1\in S$, the
clock is initially at $0$, so $D_{j_1}$ writes $j_1$ into the
$F$ register and then the cyclic shift advances the clock:
$
W_{j_1}^{(r)}|0,0,0\rangle
=
|1,j_1,0\rangle.
$
For the next $r-2$ applications the clock takes the values
$1,\ldots,r-2$.  Hence none of these applications changes either
endpoint register.  After the first $r-1$ operators have acted, the
state is
$|r-1,j_1,0\rangle.$
The last operator $W_{j_r}^{(r)}$ therefore acts at clock value
$r-1$.  If $j_r\in S$, it writes $j_r$ into the $R$ register and the
cyclic shift returns the clock to $0$.  Consequently, whenever
$j_1,j_r\in S$,
\begin{equation}
\label{eq:clock-endpoint-action}
W_{j_r}^{(r)}\cdots W_{j_1}^{(r)}
|0,0,0\rangle
=
|0,j_1,j_r\rangle.
\end{equation}

If $j_1\notin S$, then the first application does not modify the
$F$ register, which remains in $|0\rangle$.  Similarly, if
$j_r\notin S$, the $R$ register remains in $|0\rangle$.  Thus, after
projecting onto
$\mathcal E_S
:=
\operatorname{span}
\{
|0,i,j\rangle:\ i,j\in S
\}$,
only those words whose first and last labels both belong to $S$
survive.

Now evaluate the homogeneous operator $B_r$ on the test tuple
$W^{(r)}
=
\bigl(W_1^{(r)},\ldots,W_L^{(r)}\bigr).$
We have
$B_r(W^{(r)})
=
\sum_{|w|=r}
b_w W_w^{(r)},$
where, for $w=(j_1,\ldots,j_r)$,
$W_w^{(r)}
=
W_{j_r}^{(r)}\cdots W_{j_1}^{(r)}.$
Applying this operator to $|0,0,0\rangle$ and projecting onto
$\mathcal E_S$, equation \eqref{eq:clock-endpoint-action} gives
\begin{align*}
\Pi_{\mathcal E_S}
B_r(W^{(r)})
|0,0,0\rangle
&=
\sum_{ w=(j_1,\ldots,j_r)\in[L]^r }
b_w
|0,j_1,j_r\rangle
\\
&=
\sum_{i,j\in S}
\left(
\sum_{\substack{w=(j_1,\ldots,j_r)\in[L]^r \\ j_1=i, j_r=j}}
b_w
\right)
|0,i,j\rangle
\\
&=
\sum_{i,j\in S}
C^{(r)}_{ji}|0,i,j\rangle.
\end{align*}
Hence
\[
\left\|
\Pi_{\mathcal E_S}
B_r(W^{(r)})
|0,0,0\rangle
\right\|^2
=
\sum_{i,j\in S}
|C^{(r)}_{ji}|^2
=
\|C^{(r)}_{S,S}\|_F^2.
\]
Since orthogonal projection and evaluation on a unit vector cannot
increase the norm,
$\|C^{(r)}_{S,S}\|_F
\le
\|B_r(W^{(r)})\|.$
By Lemma~\ref{lem:affine-homogeneous-extraction}, the tuple
$W^{(r)}$ is a legal unitary test family and therefore
$\|B_r(W^{(r)})\|
\le \varepsilon.$
We conclude that $\|C^{(r)}_{S,S}\|_F\le\varepsilon$.
\end{proof}

\begin{lem}
\label{lem:bound of x,y}
For $x_j,y_j$ given in \eqref{eq:x and y},
we have $\sum_{j=1}^T |x_j|^2 \leq 1$ and $\sum_{j=1}^T |y_j|^2 \leq 1$.
\end{lem}

\begin{proof}
Denote $r_1=\ket{v}$ and $r_j=A_{j-1}\cdots A_1 \ket{v}$, then $|y_j|^2=\|r_j\|^2-\|A_j r_j\|^2 = \|r_j\|^2-\|r_{j+1}\|^2$. So $\sum_{j=1}^T |y_j|^2 = \|r_1\|^2 - \|r_{T+1}\|^2 \leq \|r_1\|^2 =1$. A similar argument shows $\sum_{j=1}^T |x_j|^2 \leq 1$.
\end{proof}

\begin{lem} 
\label{lem:affine-saturation}
Denote
\[
c_j=\sum_{\alpha_t=j} x_t y_t,
\quad
\beta_j=\sum_{\alpha_t=j} |x_t|^2,
\quad
\gamma_j=\sum_{\alpha_t=j} |y_t|^2,
\quad
s_j=\sqrt{\beta_j \gamma_j}.
\]
Assume $\|B(P) - A(P) \| \leq \varepsilon$ for any unitaries $P=(P_1,\ldots,P_L)$, then
\beas
\sum_{j=1}^L |c_j-\lambda_j|^2 \leq \varepsilon^2,
\quad
\left| 1 - \sum_{j=1}^L c_j \right| \leq \varepsilon, 
\quad
\sum_{j=1}^L (s_j-|c_j|) \leq  \varepsilon,
\quad
\sum_{j=1}^L (\sqrt{\beta_j}-\sqrt{\gamma_j})^2\leq 2 \varepsilon.
\eeas
\end{lem}

\begin{proof}
Using the notation of \eqref{eq: B in fourier expansion}, we have that the degree-one part is $B_1(P) =\sum_{j=1}^L c_j P_j.$
By the degree-one conclusion of
Lemma~\ref{lem:affine-homogeneous-extraction}, for every unitary tuple
$P_1,\ldots,P_L$, 
we have $\|
\sum_{j=1}^L
(c_j-\lambda_j)P_j
\|
\le \varepsilon.$
By choosing $P_j$ such that $P_j\ket{0}=\ket{j}$, we then have $\sum_j |c_j-\lambda_j|^2 \leq \varepsilon^2$.
Furthermore, if we choose all $P_j=I$, we then have $|\sum_j (c_j-\lambda_j)|\leq \varepsilon$. We then obtain the second inequality by noting that $\sum_j \lambda_j=1$.
For each fixed logical
label $j$, Cauchy-Schwarz gives
$|c_j|\leq s_j$.
By Cauchy-Schwarz once more and Lemma \ref{lem:bound of x,y} gives
\[
\sum_{j=1}^L s_j \leq \left(\sum_{j=1}^L \beta_j\right)^{1/2}
\left(\sum_{j=1}^L \gamma_j\right)^{1/2}
=
\left(\sum_{j=1}^T |x_j|^2\right)^{1/2}
\left(\sum_{j=1}^T |y_j|^2\right)^{1/2} \leq 1.
\]
On the other hand, the second inequality shows that $\sum_{j=1}^L |c_j|\geq |\sum_{j=1}^L c_j|\geq 1-\varepsilon$. Thus we obtain $\sum_{j=1}^L (s_j-|c_j|) \leq  \varepsilon$.
Finally, using $s_j=\sqrt{\beta_j\gamma_j}$, we have
$\sum_j(\sqrt{\beta_j}-\sqrt{\gamma_j})^2 = 
\sum_j\beta_j+\sum_j\gamma_j-2\sum_j s_j$.
Since $\sum_j\beta_j\leq 1, \sum_j\gamma_j\leq 1$ and 
$\sum_j s_j \geq \sum_j |c_j|\geq 1-\varepsilon$, we obtain the last inequality.
\end{proof}

We now introduce a matrix that encodes the transitions between different query positions.
Let $H\in \mathbb{C}^{T\times T}$ be a strict lower-triangular matrix with $H_{ij}=h_{ij}$ for $i>j$ and $H_{ij}=0$ for $i\leq j$. Then $H^T=0$. 
The preceding lemmas allow us to convert the approximation condition into quantitative constraints on the query histories.

\begin{lem} 
\label{lem:query-history-estimates}
Fix $\tau>\varepsilon$. Let
$
r={\varepsilon}/{\tau}<1,
$ and
$
q_r=\frac{\sqrt{4-2r}-\sqrt{2r}}{2}.
$
Assume $N\ge T(M_\tau+1)^2$. Then there exist isometries
$U,V:\mathbb C^{M_\tau}\to\mathbb C^T$ such that, for $1\le k<T$,
$$
\|U^\dagger H^kV\|_F
\le
\frac{r}{q_r^2},
\quad 
\|U-V\|_F
\le
\sqrt{\frac{2r}{1-r}}.
$$

\end{lem}

\begin{proof}
We use the notation from Lemma~\ref{lem:affine-saturation}.
For every $j\in J_\tau$,
$
|c_j-\lambda_j|\le\varepsilon,
$
and hence
$
|c_j|
\ge
\lambda_j-\varepsilon
\ge
\tau-\varepsilon.
$
On the other hand, by Cauchy--Schwarz,
$
|c_j|\le s_j=\sqrt{\beta_j\gamma_j},
$
so
$
\sqrt{\beta_j\gamma_j}\ge\tau-\varepsilon.
$
Moreover, Lemma~\ref{lem:affine-saturation} gives
$
|\sqrt{\beta_j}-\sqrt{\gamma_j}|
\le
\sqrt{2\varepsilon}.
$
Set
$
a_j:=\min\{\sqrt{\beta_j},\sqrt{\gamma_j}\}.
$
Then
$
\max\{\sqrt{\beta_j},\sqrt{\gamma_j}\}
\le
a_j+\sqrt{2\varepsilon},
$
and therefore
$
a_j(a_j+\sqrt{2\varepsilon})
\ge
\tau-\varepsilon.
$
Since $q\mapsto q(q+\sqrt{2\varepsilon})$ is increasing for $q\ge0$, and
$
\sqrt{\tau}\,q_r
(\sqrt{\tau}\,q_r+\sqrt{2\varepsilon})
=
\tau-\varepsilon,
$
we obtain
$
a_j\ge\sqrt{\tau}\,q_r.
$
Consequently,
\begin{equation}
\label{eq:alpha-beta-lower}
\beta_j,\gamma_j\ge\tau q_r^2,
\qquad j\in J_\tau.
\end{equation}

For each $j\in J_\tau$, define
$$
\ket{u_j}
=
\frac{1}{\sqrt{\beta_j}}
\sum_{\alpha_t=j}\bar{x}_t\ket{t},
\qquad
\ket{v_j}
=
\frac{1}{\sqrt{\gamma_j}}
\sum_{\alpha_t=j}y_t\ket{t}.
$$
Because distinct logical labels have disjoint occurrence supports,
the families $\{\ket{u_j}:j\in J_\tau\}$ and
$\{\ket{v_j}:j\in J_\tau\}$ are both orthonormal.
Thus they define two isometries
$
U,V:\mathbb C^{M_\tau}\to\mathbb C^T.
$
Let
$
D_\beta=\operatorname{diag}(\beta_j)_{j\in J_\tau},$ and $
D_\gamma=\operatorname{diag}(\gamma_j)_{j\in J_\tau}.
$
For every $1\le k<T$, the path expansion gives
\begin{equation}
\label{eq:endpoint-factorization}
C^{(k+1)}_{J_\tau,J_\tau}
=
D_\beta^{1/2}
U^\dagger H^kV
D_\gamma^{1/2}.
\end{equation}
Indeed, the $(j,i)$-entry of the right-hand side is
$
\sqrt{\beta_j\gamma_i}\,
\langle u_j,H^kv_i\rangle,
$
which is exactly the sum of all length-$(k+1)$ occurrence paths whose first logical label is $i$ and whose last logical label is $j$.
By Lemma~\ref{lem:principal submatrix},
$
\|C^{(k+1)}_{J_\tau,J_\tau}\|_F\le\varepsilon.
$
The assumption $N\ge T(M_\tau+1)^2$ is sufficient since $k+1\le T$.
Using \eqref{eq:alpha-beta-lower},
$
\|D_\beta^{-1/2}\|,
\|D_\gamma^{-1/2}\|
\le
\frac{1}{\sqrt{\tau}\,q_r}.
$
Hence
\begin{equation}
\label{eq:bound 1}
\|U^\dagger H^kV\|_F
\le
\frac{\varepsilon}{\tau q_r^2}
=
\frac{r}{q_r^2}.
\end{equation}
By the definitions of $\ket{u_j}$ and $\ket{v_j}$,
$
\langle u_j,v_j\rangle={c_j}/{s_j}.
$
Thus
$
U^\dagger V
=
\operatorname{diag}(c_j/s_j)_{j\in J_\tau}.
$
Multiplying each column $\ket{v_j}$ by a suitable phase preserves orthonormality and the Frobenius-norm bound in \eqref{eq:bound 1}, so we may assume
$
\langle u_j,v_j\rangle
=
{|c_j|}/{s_j}
\ge0.
$
Therefore,
$
\|U-V\|_F^2
=
2\sum_{j\in J_\tau}
{(s_j-|c_j|)}/{s_j}.
$
Since
$
s_j\ge|c_j|\ge\tau-\varepsilon
$
for $j\in J_\tau$, and Lemma~\ref{lem:affine-saturation} gives
$
\sum_{j=1}^L(s_j-|c_j|)
\le\varepsilon,
$
we obtain
\[
\|U-V\|_F
\le
\sqrt{\frac{2\varepsilon}{\tau-\varepsilon}}
=
\sqrt{\frac{2r}{1-r}}.
\qedhere
\]
\end{proof}

\subsection{From query-history constraints to the ancilla lower bound}

Using the notation from Lemma \ref{lem:reformulation}, we define the strict lower-triangular overlap matrix $\Gamma$ via
\[
\Gamma_{ij} = \braket{\psi_i}{\psi_j}~~ (i>j),
\quad 
\Gamma_{ij} = 0~~(i\leq j).
\]
Since $A_k+\Pi_k=I$, we have $\braket{\psi_i}{\psi_j} = \bra{\psi_i} \prod_{k=i-1}^{j+1}(A_k+\Pi_k) \ket{\psi_j}$. Expanding this gives
\[
\Gamma=H+H^2+\cdots H^{T-1}, \quad I+\Gamma=(I-H)^{-1}.
\]

\begin{lem}
\label{lem:gram matrix}
Let $G$ be the Gram matrix of $\ket{\psi_1},\ldots,\ket{\psi_T}\in (\mathbb{C}^2)^{\otimes a}$, then
\[
G = (I-H)^{-\dag} (I-H^\dag H) (I-H)^{-1}.
\]
Consequently, $I-H^\dag H \succeq 0, \|H\|\leq 1$ and $\rank(I-H^\dag H)=\rank(G) \leq 2^a$.
\end{lem}

\begin{proof}
Denote $R=I+\Gamma=(I-H)^{-1}$, then $G=I+\Gamma+\Gamma^\dag=R+R^\dag-I$. Since $HR=R-I$, we then obtain $R^\dag (I-H^\dag H)R = R^\dag R-(R^\dag-I)(R-I)=G.$
\end{proof}

Denote $D=(I-H^\dag H)^{1/2}$.
From Lemma \ref{lem:gram matrix}, $D$ is well-defined, and $d=\rank(D)\leq 2^a$.
We will identify the range of $D$ as $\mathbb{C}^d$. Let
$\mathcal{T}$ be a Hilbert space with orthonormal basis $\{\ket{j}:j\geq0\}$.
Define
linear operator $\mathcal{J}$ from $\mathbb{C}^T$ to $\mathbb{C}^d\otimes \mathcal{T}$ as 
\be \label{eq:finite-defect-shift}
\mathcal{J}x := \sum_{j=0}^{T-1} D H^j x \otimes \ket{j}
\ee
Then $\mathcal{J}$ is isometry. Indeed, since $H^T=0$, we have
\[
\|\mathcal{J} x\|^2 = \sum_{j=0}^{T-1} \langle x, (H^j)^\dag (I-H^\dag H) H^j x \rangle
= \langle x, (I-(H^T)^\dag H^T) x \rangle = \|x\|^2.
\]
Let $S$ be the forward shift operator acting on $\mathcal{T}$ defined by $S \ket{j}=\ket{j+1}$. Then it is easy to check that $\mathcal{J} H = (I\otimes S^\dag) \mathcal{J}$.

\begin{lem}
\label{lem:defect-shift-packing}
Let $U:\mathbb{C}^M \rightarrow \mathbb{C}^T$ be an isometry and suppose $\|U^\dag H^k U\|_F \leq \kappa$ for $1\leq k \leq T-1$, then
\[
2^a \geq \frac{M^2}{M+2\kappa^2(T-1)}.
\]
\end{lem}

\begin{proof}
Let $u_j$ be the $j$-th column of $U$ for $j=1,\ldots,M$.
They are orthonormal.
Let $\mathcal{J}$ be the isometry from \eqref{eq:finite-defect-shift}.
Define $f_j:=\mathcal J u_j$, then the vectors $f_1,\ldots,f_M$ are also orthonormal.
Since $\mathcal{J} H = (I\otimes S^\dag) \mathcal{J}$, iterating the intertwining relation gives 
$\mathcal J H^k
=(I\otimes S^\dagger)^k\mathcal J$ for $0\leq k < T$.
Hence, $\langle (I\otimes S^k) f_j,f_i\rangle
=
\langle u_j,H^k u_i\rangle.$
Therefore,
\begin{equation}
\label{eq:shift-correlation-bound}
\sum_{i,j=1}^M
|\langle (I\otimes S^k) f_j,f_i\rangle|^2
=
\|U^\dagger H^k U\|_F^2
\le
\kappa  ^2,
\qquad
1\le k<T.
\end{equation}
Fix an integer \(K\ge1\), and consider the family
$\mathcal F_K
:=
\{
(I\otimes S^p) f_j:
0\le p<K,\;
1\le j\le M
\}.$
There are $MK$ vectors in \(\mathcal F_K\).
Since each \(f_j\) is supported on \(T\) consecutive time coordinates,
the shifted vectors \(S^p f_j\), \(0\le p<K\), are supported within
the time interval $0,\ldots,T+K-2.$
Thus the whole family $\mathcal F_K$ lies in $\mathbb C^d\otimes\mathbb C^{T+K-1},$
whose dimension is $d(T+K-1)\le 2^a(T+K-1)$.
If \(G_K\) denotes the Gram matrix of \(\mathcal F_K\), it follows that
\begin{equation}
\label{eq:gram-rank-upper}
\rank(G_K)
\le
2^a(T+K-1).
\end{equation}
We next derive a lower bound on $\rank(G_K)$.
The entries of $G_K$ are
$(G_K)_{(p,j),(q,i)}
=
\langle (I\otimes S^p) f_j, (I\otimes S^q) f_i\rangle.$
Because each $(I\otimes S^p) f_j$ is a unit vector, $\operatorname{Tr}(G_K)=MK.$
Moreover,
\[
\operatorname{Tr}(G_K^2)
=
\sum_{p,q=0}^{K-1}
\sum_{i,j=1}^M
\left|
\langle (I\otimes S^p) f_j, (I\otimes S^q) f_i\rangle
\right|^2.
\]
We separate the contributions according to the time difference $k:=|p-q|.$
When $p=q$, orthonormality of \(f_1,\ldots,f_M\) gives
\[
\sum_{i,j=1}^M
|\langle (I\otimes S^p) f_j, (I\otimes S^q) f_i \rangle|^2
=
M.
\]
Summing over \(p=0,\ldots,K-1\), the diagonal contribution is \(MK\).
For \(k\ge1\), there are exactly \(K-k\) ordered pairs
\((p,q)\) with \(q=p+k\), and the same number with \(p=q+k\).
Furthermore, if \(k\ge T\), the corresponding shifted supports are
disjoint, so the inner products vanish.  Hence
\[
\operatorname{Tr}(G_K^2)
=
MK
+
2
\sum_{k=1}^{\min\{K-1,T-1\}}
(K-k)
\sum_{i,j=1}^M
|\langle (I\otimes S^k)  f_j,f_i\rangle|^2.
\]
Using \eqref{eq:shift-correlation-bound},
\[
\operatorname{Tr}(G_K^2)
\le
MK
+
2\kappa^2
\sum_{k=1}^{\min\{K-1,T-1\}}
(K-k)
\le
MK
+
2\kappa^2 K(T-1).
\]
Since \(G_K\succeq0\), if
$\mu_1,\ldots,\mu_R$ 
are its nonzero eigenvalues, where $R=\operatorname{rank}G_K$,
then Cauchy--Schwarz gives
$(
\sum_{\ell=1}^R\mu_\ell)^2
\le
R
\sum_{\ell=1}^R\mu_\ell^2.$
Equivalently,
$(\operatorname{Tr}G_K)^2
\le
\operatorname{rank}(G_K)\operatorname{Tr}(G_K^2).$
Therefore
\begin{align}
\operatorname{rank}(G_K)
\ge
\frac{(\operatorname{Tr}G_K)^2}
{\operatorname{Tr}(G_K^2)}
\ge
\frac{(MK)^2}
{MK+2\kappa^2K(T-1)}
=
\frac{M^2K}
{M+2\kappa^2(T-1)}.
\label{eq:gram-rank-lower}
\end{align}
Combining
\eqref{eq:gram-rank-upper} and
\eqref{eq:gram-rank-lower}, we obtain
\[
2^a(T+K-1)
\ge
\frac{M^2K}
{M+2\kappa^2(T-1)}.
\]
Dividing by \(K\) and letting $K\rightarrow \infty$ yields the claim.
\end{proof}

We are now ready to prove Theorem~\ref{thm:General effective-support tradeoff}.

\begin{proof}[Proof of Theorem~\ref{thm:General effective-support tradeoff}]
Fix $\tau>\varepsilon$ and let
$
r={\varepsilon}/{\tau}<1.
$
Define
$$
q_r=\frac{\sqrt{4-2r}-\sqrt{2r}}{2},
\qquad
K_r=
\frac{r}{q_r^2}
+
\sqrt{\frac{2r}{1-r}}.
$$
By Lemma~\ref{lem:query-history-estimates}, there exist isometries
$U,V:\mathbb C^{M_\tau} \to\mathbb C^T$ such that, for every $1\le k<T$,
$$
\|U^\dagger H^kV\|_F
\le
\frac{r}{q_r^2},
\qquad
\|U-V\|_F
\le
\sqrt{\frac{2r}{1-r}}.
$$
By Lemma~\ref{lem:gram matrix}, $\|H\|\le1$. Hence
$$
\|U^\dagger H^kU\|_F
\le
\|U^\dagger H^kV\|_F
+
\|U^\dagger H^k(U-V)\|_F
\le
\frac{r}{q_r^2}
+
\sqrt{\frac{2r}{1-r}}
=
K_r.
$$
Thus
$
\|U^\dagger H^kU\|_F\le K_r,$ for $
1\le k<T.
$
Applying Lemma~\ref{lem:defect-shift-packing} with
$\kappa=K_r$ gives
\[
2^a
\ge
\frac{M_\tau^2}
{M_\tau+2K_{\varepsilon/\tau}^2(T-1)}.
\qedhere
\]
\end{proof}

\section{An optimal quantum algorithm for normalized trace estimation in DQC1}
\label{sec:applications}

We next describe an application of our low-ancilla construction of block-encodings in \cref{thm:MPF} 
to normalized trace estimation. We first provide an improved quantum algorithm over \cite{cade2018quantum} in terms of approximate degree, then prove a matching lower bound to show our algorithm is optimal.

\subsection{An optimal quantum algorithm for normalized trace estimation}

Given an $n$-qubit Hamiltonian $A$ and a function $f(x)$, the goal is to estimate  $\frac{1}{2^n}\operatorname{Tr}[f(A)].$
This problem is a natural task in the DQC1 model, where normalized traces of unitaries can be estimated using one clean qubit~\cite{knill1998power,shor2008estimating}.
Normalized trace estimation for functions of log-local Hamiltonians was recently studied in~\cite{cade2018quantum,edenhofer2025dequantization,ji2026dqc1} and the latter work identifies the approximate degree, defined below, of $f(x)$ as a key parameter governing the complexity of this problem.

\begin{defn}[Approximate degree]
\label{defn:approximate degree}
Let $f(x): [-1,1] \rightarrow [-1,1]$ be a function and $\varepsilon\in(0,1]$, we define the approximate degree of $f(x)$ as
\be \label{app deg}
\widetilde{\deg}_{\varepsilon}(f) = \min\{d \in \mathbb{N}: |f(x) -g(x)| \leq \varepsilon 
\text{ for all } x\in [-1,1], \deg(g)\leq d \}.
\ee
\end{defn}

A standard algorithmic approach for $\frac{1}{2^n}\operatorname{Tr}[f(A)]$ is based on quantum phase estimation. Specifically, Cade and Montanaro showed that, for a log-local Hamiltonian $A$ and a Lipschitz-continuous function $f(x)$, one can efficiently estimate $\frac{1}{2^n} \operatorname{Tr}[f(A)]$ in DQC1 with accuracy depending on the Lipschitz constant~\cite[Corollary 6]{cade2018quantum}. We restate as follows.

\begin{prop}
\label{prop:dqc1-qpe}
Assume that $A$ is a log-local Hamiltonian acting on $n$ qubits and $f(x):[-1,1]\rightarrow [-1,1]$ is Lipschitz continuous with Lipschitz constant $K$, then there is a DQC1 algorithm to estimate $\frac{1}{2^n} \tr[f(A)] \pm \eps (K+1)$, where $\eps=1/O({\rm poly}(n))$.
\end{prop}

The cost of this algorithm is linear in $1/\eps$ from quantum phase estimation. If we set the final error as $\varepsilon$, then the overall cost is $\widetilde{O}(K/\varepsilon)$. 
Applying the phase-estimation algorithm to a bounded degree-\(d\) polynomial approximant of \(f\), Markov brothers' inequality gives a Lipschitz constant \(O(d^2)\). Thus this phase-estimation method may lead to a quadratic dependence on the approximate degree.

To state our improved quantum algorithm, we need the following two lemmas.


\begin{lem}[Theorem 30 of \cite{gilyen2019qsvt}, rephrased] 
\label{lem:amplified-block-encoding}
Let $\gamma>1$ and $\delta, \varepsilon\in(0,1/2)$. Suppose that $U$ is a $(1,a,0)$-block encoding of an operator $M$ satisfying $ \|M\|\leq\frac{1-\delta}{\gamma}.$ Then there is a quantum circuit that implements a \((1,a+1,\varepsilon)\)-block encoding of \(\gamma M\), using $ O(\frac{\gamma}{\delta}\log\frac{\gamma}{\varepsilon})$ applications of \(U\) and \(U^\dagger\).
\end{lem}

\begin{lem}
\label{lem:amplified-be}
Let \(\gamma>1\) and \(\delta,\eta\in(0,1/2)\).
Suppose that \(U\) is a \((\gamma,a,\eta)\)-block encoding of a Hermitian operator \(A\) satisfying $\|A\|\leq1-\delta.$ Then there is a quantum circuit that implements a \((1,a+1,2\eta)\)-block encoding of \(A\), using $ O(\frac{\gamma}{\delta}\log\frac{\gamma}{\eta})$ applications of \(U\) and \(U^\dagger\).
\end{lem}


\begin{proof}
Define $M=(\langle0^a|\otimes I)U(|0^a\rangle\otimes I)$, then we have $\|\gamma M-A\|\leq\eta$. If \(\gamma\leq1+\eta\), then
$$ \|M-A\| \leq\|\gamma M-A\|+(\gamma-1)\|M\| \leq2\eta, $$
hence \(U\) already gives the required block encoding. Otherwise, let $\tilde{\gamma}=\frac{\gamma}{1+\eta} >1$. Then
\[
\|\tilde{\gamma} M\| \leq \frac{1-\delta+\eta}{1+\eta}=1-\frac{\delta}{1+\eta} \leq 1-\delta / 2 .
\]
Applying \cref{lem:amplified-block-encoding} with accuracy \(\eta^2\) yields a unitary whose top-left block \(B\) satisfies $\|B-\tilde\gamma M\|\leq\eta^2.$ Hence, we have
\[
\begin{aligned}
    \|B-A\| &\leq\|B-\tilde\gamma M\|+\|\tilde\gamma M-A\| \\
    &\leq \eta^2+\frac{\eta(1+\|A\|)}{1+\eta}
    \leq2\eta. 
\end{aligned}
\]
The number of applications of $U$ and $U^\dag$ is $O(\frac{\tilde{\gamma}}{\delta} \log \frac{\tilde{\gamma}}{\eta^2})=O(\frac{\gamma}{\delta} \log \frac{\gamma}{\eta})$.
\end{proof}


Now we state the resulting DQC1 algorithm for normalized trace estimation.

\begin{thm}
\label{thm:trace-qsvt}
Assume that $A=\sum_{j=1}^L \alpha_j H_j$ is a Hermitian decomposition with $\alpha_j \in \mathbb{R}_{\geq 0}$ and $\|H_j\|=1$. Let $\alpha=\sum_j \alpha_j$, and
assume that the time evolution of each $H_j$ costs $O(1)$. 
Suppose that $\|A\|\le 1-\delta$ for some constant $\delta\in(0,1/2)$. 
 Let $f(x):[-1,1]\to[-1,1]$ be continuous.
Then there is a DQC1 algorithm that estimates 
$\frac{1}{2^n}\operatorname{Tr}[f(A)]$
to additive error $\eps$, with circuit depth
\begin{equation}
    \widetilde{O} \left( \alpha L \cdot  \widetilde{\deg}_{\varepsilon/3}(f) \right)
\end{equation}
per run and the total cost $\widetilde O( \alpha L \cdot  \widetilde{\deg}_{\varepsilon/3}(f) /\varepsilon^2 ),$
and uses $O(\log\log(\alpha \cdot \widetilde{\deg}_{\varepsilon/3}(f) /\eps))$ clean qubits.
\end{thm}

\begin{proof}
The MPF-based construction in~\cref{thm:MPF} gives an $(\alpha\pi/2,O(\log\log(\alpha/\eta)),\eta)$-block encoding of $A$, with circuit depth $\widetilde O(L)$. 
We first reduce the block-encoding factor $\alpha\pi/2$ to $1$. 
Applying~\cref{lem:amplified-be} with $\gamma=\alpha\pi/2$ gives a $(1,O(\log\log(\alpha/\eta)),2\eta)$-block encoding of $A$.
Since $\delta$ is a constant, this uses $\widetilde O(\alpha)$ applications of the original block encoding.
Hence the resulting block encoding of $A$ has circuit depth $\widetilde O(\alpha L).$

Let $P(x)\in\mathbb R[x]$ be a polynomial of degree $d=\widetilde{\deg}_{\varepsilon/3}(f)$ that approximates $f$ such that $\|P-f\|_{[-1,1]} \leq \eps/3$.
Applying QSVT to the above block encoding and the polynomial $P$ yields a unitary $U_P$ whose top-left block $B=(\langle 0^b|\otimes I)U_P(|0^b\rangle\otimes I)$ satisfies $\|B-P(A)\|\le \eps/3,$
by setting $\eta=\Theta({\eps^2}/{d^2}).$
This uses $b=O(\log\log(\alpha/\eta))=O(\log\log(\alpha d/\varepsilon))$ ancilla qubits and has circuit depth $\widetilde O(\alpha Ld)$.

The normalized trace of $B$ can then be estimated using the usual DQC1 trace-estimation procedure.
Starting from 
$|+\rangle\langle+|\otimes |0^b\rangle\langle 0^b|\otimes \frac{I_{2^n}}{2^n},$
we can apply controlled-$U_P$, controlled by the first qubit being 1, and measure $X$ on the control qubit. The expectation value is $\operatorname{Re}(2^{-n}\operatorname{Tr}[B])$. For real polynomials and Hermitian $A$, the trace is real. Therefore, this procedure estimates $2^{-n}\operatorname{Tr}[P(A)]$, and hence $2^{-n}\operatorname{Tr}[f(A)]$. Moreover,
\[
\left|
\frac{1}{2^n}\operatorname{Tr}[B]
-
\frac{1}{2^n}\operatorname{Tr}[P(A)]
\right|
\le \|B-P(A)\|
\le \eps/3,
\]
and since $\|P-f\|_{[-1,1]}\le \eps/3$, we have
\[
\left|
\frac{1}{2^n}\operatorname{Tr}[P(A)]
-
\frac{1}{2^n}\operatorname{Tr}[f(A)]
\right|
\le \eps/3.
\]
Repeating the procedure $O(1/\eps^2)$ times gives sampling error at most $\eps/3$ with constant probability. Combining these gives an estimate of $2^{-n} \tr[f(A)]$ to additive error $\eps$.
\end{proof}

\begin{rmk}
    This method uses $O(\log\log(\alpha d/\eps))$ clean qubits in addition to the DQC1 control qubit. When $\alpha,d,L=\operatorname{poly}(n)$ and the additive error $\eps=1/{\rm poly}(n)$, the algorithm has polynomial total cost and uses $O(\log\log n)$ clean qubits. Hence, it lies in the standard DQC$k$ framework for $k=O(\log n)$, which is equivalent to DQC1~\cite{shor2008estimating}.
\end{rmk}

Thus, for log-local Hamiltonians, or more generally for Hermitian decompositions with efficiently implementable term evolutions, our construction gives a low-ancilla QSVT approach to normalized trace estimation. Compared with the approach based on phase estimation, whose cost depends on the Lipschitz constant of $f(x)$, this method depends linearly on the approximate degree of $f(x)$ up to logarithmic factors. 

\subsection{Lower bound analysis}

We next show that the degree dependence in the above algorithm is essentially optimal in a natural DQC1 query model.
We consider the following special type of Jacobi matrices, which is known as {\em periodic Jacobi matrices}, e.g., see \cite{van1976spectrum, ferguson1980construction}:
\be \label{eq:periodic Jacobi matrix}
A_\theta :=\begin{bmatrix}
a_1    & b_1 &      & & b_d e^{i\theta} \\
b_1  & a_2   & b_2  & \\
     & b_2 &  \ddots   & \ddots \\
     &     & \ddots        & \ddots & b_{d-1}\\
b_d e^{-i\theta}        &&  & b_{d-1} & a_d \\
\end{bmatrix},
\ee
where $a_1,\ldots,a_d \in \mathbb{R}$, 
$b_1,\ldots,b_d \in \mathbb{R}^+$ and $\theta\in[0,2\pi)$. By construction, $A_\theta$ is Hermitian, and hence all the eigenvalues are real.
Its characteristic polynomial satisfies
\be \label{char poly of A}
\det(xI-A_\theta) = 
2b_1b_2\cdots b_d (\Delta(x) -  \cos\theta).
\ee
The polynomial $\Delta(x)$ 
is known as the {\em discriminant} of $A_\theta$ and plays an important role in the study of $A_\theta$. It is easy to see that $ \Delta(A_\theta)= (\cos\theta ) I_d$.

\begin{defn}[DQC1 normalized-trace query problem]
Assume $x\in \{0,1\}^d$ is a bit string with query oracle $O_x$.
Let $x\mapsto A_x$ be a family of Hermitian matrices of dimension $N$ with spectrum contained in the domain of $f$.  Define
$
\tau_f(x):=\frac{1}{N}\tr[f(A_x)].
$
For $\delta>0$, let
$
\mathsf{Q}^{\rm DQC1}_{\delta}(\tau_f)
$
denote the smallest number of queries to $O_x$ used by a DQC1 oracle algorithm that outputs an estimate of $\tau_f(x)$ to additive error $\delta$ with probability at least $2/3$ for every valid input $x$. We will ignore $\delta$ when it is a constant.
\end{defn}

This is a query model for a promise family $A_x$. The algorithm is not assumed to query entries of $A_x$; it only queries $O_x$. In the DQC1 query model, for functions of matrices, it is not reasonable to consider the sparse-access query model because it uses too many ancilla qubits. The above is one way to define a query problem for functions of matrices in the DQC1 model.

Let $a_1,\ldots,a_d\in\mathbb{R}$ and $b_1,\ldots,b_d\in \mathbb{R}^+$ be fixed coefficients.
Define the periodic Jacobi matrix associated to $x$ as
\begin{equation}\label{eq:Ax-cycle}
    A_x=
    \sum_{j=1}^d a_j\ket{j}\bra{j}
    +\sum_{j=1}^d b_j(-1)^{x_j}
    \left(\ket{j}\bra{j+1}+\ket{j+1}\bra{j}\right),
\end{equation}
where $\ket{d+1}:=\ket{1}$. Then $A_x$ is unitarily equivalent to
\begin{equation}\label{eq:Jz}
    J_{x}=\sum_{j=1}^d a_j\ket{j}\bra{j}
    +\sum_{j=1}^{d-1} b_j
    \left(\ket{j}\bra{j+1}+\ket{j+1}\bra{j}\right)
    + b_d (-1)^{ x_1\oplus\cdots\oplus x_d}
    \left(\ket{d}\bra{1}+\ket{1}\bra{d}\right).
\end{equation}
Consequently, $\tr [f(A_x)]=\tr [f(J_{x})]$
for every function $f$ defined on the spectrum.

\begin{defn}
Let $f(x): [a,b] \subseteq [-1,1]\rightarrow \mathbb{R}$ be a continuous function. For any integer $d$, define
\be \label{eq for Em}
E_d(f;[a,b]):=\min_{\deg(P)\leq d} \,\, \max_{x\in [a,b]} \,\, |f(x) - P(x)|
\ee
to be the best uniform (minimax) polynomial approximation error of degree at most $d$ over $[a,b]$. When there is no risk of confusion, we simply write $E_d$.
The polynomial of degree at most $d$ that attains this minimum is denoted by 
$P_d^*(x)$. 
\end{defn}

The proof of \cite[Theorem 2.10]{ji2026dqc1} implies the following result.

\begin{prop}
\label{prop:periodic-witness}
Let $d=\widetilde{\deg}_{\varepsilon}(f)$ and let $\eta\in[E_d,E_{d-1})$.  Let $P_{d-1}^*$ be a best degree-$(d-1)$ minimax approximation to $f(x)$.  Then there exists a $d$-dimensional periodic Jacobi matrix
\[
    A_\theta=
    \sum_{j=1}^d a_j\ket{j}\bra{j}
    +\sum_{j=1}^{d-1}b_j(\ket{j}\bra{j+1}+\ket{j+1}\bra{j})
    +b_d(e^{-i\theta}\ket{d}\bra{1}+e^{i\theta}\ket{1}\bra{d}),
\]
with $\|A_\theta\|\le 1$ and discriminant $\Delta$ such that for some $\sigma\in\{\pm 1\}$
\[
    f(x)=P_{d-1}^*(x)+\sigma(E_{d-1}-\eta)\Delta(x)+r(x),
    \qquad
    \|r\|_\infty\le\eta.
\]
\end{prop}

\begin{lem}\label{lem:trace-gap}
Under the assumptions of Proposition \ref{prop:periodic-witness}, we have
\[
    \gamma:=
    \left|
    \frac{1}{d} \tr [f(A_0)]-\frac{1}{d} \tr [f(A_\pi)]
    \right|
    \ge 2(E_{d-1}-\eta)-2\eta
    =2E_{d-1}-4\eta.
\]
In particular, taking $\eta=E_d$ gives
$
    \gamma\ge 2E_{d-1}-4E_d.
$
\end{lem}

\begin{proof}
Using Proposition \ref{prop:periodic-witness}, $f=P_{d-1}^*+\sigma(E_{d-1}-\eta)\Delta+r$ with $ \|r\|_\infty\le\eta$.
Note that $\deg(P_{d-1}^*)\leq d-1<d$, we have
$\tr [P_{d-1}^*(A_0)]=\tr [P_{d-1}^*(A_\pi)]$
since they are independent of $\theta$.
Also $\Delta(A_0)=I_d, \Delta(A_\pi)=-I_d.$
Therefore
\[
    \frac{1}{d}\tr [f(A_0)]-\frac{1}{d}\tr [f(A_\pi)]
    =2\sigma(E_{d-1}-\eta)
    +\frac{1}{d}\tr [r(A_0)]-\frac{1}{d}\tr [r(A_\pi)].
\]
Since the spectra of $A_0, A_\pi$ lie in $[-1,1]$, we have $|\frac{1}{d}\tr [r(A_\theta)]|\le \eta.$
The claim follows.
\end{proof}

Let $A_0,A_\pi$ be the two boundary conditions supplied by Proposition \ref{prop:periodic-witness}.  Define $A_x$ by \eqref{eq:Ax-cycle} using the same coefficients $a_j,b_j$. By \eqref{eq:Jz}, it is easy to see that 
\be
    A_x\simeq
    \begin{cases}
    A_0, & x_1\oplus\cdots\oplus x_d=0,\\
    A_\pi, & x_1\oplus\cdots\oplus x_d=1.
    \end{cases}
\ee
Consequently,
\be
    \frac1d\tr [f(A_x)]=
    \begin{cases}
    d^{-1}\tr [f(A_0)], & x_1\oplus\cdots\oplus x_d=0,\\
    d^{-1}\tr [f(A_\pi)], & x_1\oplus\cdots\oplus x_d=1.
    \end{cases}
\ee

\begin{prop}
\label{prop:query-gap}
Suppose the periodic Jacobi pair $A_0,A_\pi$ satisfies
\bes
    \gamma=
    \left|
    \frac1d\tr [f(A_0)]-\frac1d\tr [f(A_\pi)]
    \right|>0.
\ees
Then estimating $\tau_f(x)$
to additive error $\gamma/4$ requires $\Omega(d)$ queries.  Hence
$
    \mathsf{Q}^{\rm DQC1}_{\gamma/4}(\tau_f)=\Omega(d).
$
\end{prop}

\begin{proof}
If a quantum query algorithm estimates $\tau_f(x)$ to error $\gamma/4$, then it distinguishes the two cases $x_1\oplus\cdots\oplus x_d=0$ and $x_1\oplus\cdots\oplus x_d=1$ with constant success probability.  Thus it computes $d$-bit parity.
The quantum query complexity of computing parity is $\Theta(d)$, e.g., see \cite{beals2001quantum}, so we obtain the claimed result.
\end{proof}

\begin{thm}[Quantum query complexity lower bound]
\label{thm:approx-degree-lb}
Let $f(x):[-1,1]\to[-1,1]$ be continuous and let $d=\widetilde{\deg}_{\varepsilon}(f).$
Assume that for some constant $c<1/2$, $E_d\le cE_{d-1}.$
Then the family $A_x$ constructed above satisfies
$
\mathsf{Q}^{\rm DQC1}_{\gamma/4}(\tau_f)=\Omega(d),
$
where $\gamma\ge 2(1-2c)E_{d-1}.$
In particular, if $E_{d-1}=\Omega(1)$, then estimating $\tau_f(x)$ to constant additive error requires
\be
\mathsf{Q}^{\rm DQC1}(\tau_f)=\Omega (\widetilde{\deg}_{\varepsilon}(f) ).
\ee
\end{thm}

\begin{proof}
Apply Lemma \ref{lem:trace-gap} with $\eta=E_d$.  The assumption $E_d\le cE_{d-1}$ gives $\gamma\ge 2E_{d-1}-4E_d\ge 2(1-2c)E_{d-1}.$
The query lower bound follows from Proposition \ref{prop:query-gap}.
\end{proof}

In practice, we can choose $\varepsilon$ as a small constant. Let $d=\widetilde{\deg}_{\varepsilon}(f)$, then we have $E_{d-1}>\varepsilon\ge E_d$, so $E_{d-1}$ is also a constant. Regarding the condition $E_d\leq c E_{d-1}$, as discussed in the updated version of \cite{ji2026dqc1}, it can be relaxed to only assume that $f(x)$ is Lipschitz.

\section{Conclusion}

In this work, we studied low-ancilla block encodings from both constructive and lower-bound perspectives. We gave two constructions for Hermitian decompositions, illustrating how ancilla space can be reduced at the cost of additional circuit complexity.

Our main results characterize this tradeoff more fundamentally. For a broad coherent simulation architecture based on Suzuki product formulas, we established a quantitative ancilla-depth tradeoff, showing that limited ancilla space imposes intrinsic constraints on the achievable precision at a given circuit depth. In particular, for a constant number of coherent rounds, under the polynomial-resource assumptions of \cref{thm:tradeoff}, constant ancilla space cannot be combined with polylogarithmic dependence on the inverse precision within this architecture. We also proved an ancilla-query tradeoff for approximate LCU in a general repeated-query model, extending the exact logarithmic ancilla lower bound and quantifying how additional queries can compensate for reduced ancilla space.

These results suggest that ancilla space should be viewed as a fundamental computational resource rather than merely an implementation overhead. Understanding analogous tradeoffs in more general circuit models, beyond the architectures considered here, remains an important direction for future work.

\section*{Acknowledgements}

YZ and CS were supported by the National Key Research Project of China under Grant No. 2025YFA1017200. YZ was also supported by the European Research Council (ERC) via the Starting Grant QuanThermal (101222179).
We would like to thank Dong An, Shantanav Chakraborty, and Zane Rossi for their helpful discussions.

\paragraph{AI statement:} The tradeoff results in this paper were developed in part through interactive discussions with GPT-5.6 Sol, which was used to explore possible proof ideas and formulations. All mathematical statements and proofs were subsequently simplified and verified by the authors, who take full responsibility for their correctness.



\bibliographystyle{alphaurl}
\bibliography{qsvt}

\end{document}